\documentclass[preprint,11pt]{elsarticle}
\usepackage{graphicx}
\usepackage{amsmath,amssymb,amsthm}
\usepackage{enumitem}
\usepackage[colorlinks=true,linkcolor=blue,citecolor=blue,urlcolor=blue]{hyperref}
\usepackage{xcolor}
\usepackage{algorithm,algpseudocode}
\journal{Stochastic Processes and their Applications}

\newtheorem{theorem}{Theorem}[section]
\newtheorem{proposition}[theorem]{Proposition}
\newtheorem{lemma}[theorem]{Lemma}
\newtheorem{corollary}[theorem]{Corollary}
\theoremstyle{remark}
\newtheorem{remark}[theorem]{Remark}

\newcommand{\R}{\mathbb{R}}
\newcommand{\N}{\mathbb{N}}
\newcommand{\E}{\mathbb{E}}
\newcommand{\Prob}{\mathbb{P}}
\newcommand{\ind}{\mathbf{1}}
\newcommand{\Mp}{\mathcal{M}_p}
\DeclareMathOperator{\supp}{supp}
\newcommand{\nuY}{\nu_Y}
\DeclareMathOperator{\Cov}{Cov}
\newcommand{\PsiV}{\Psi_{V}}
\newcommand{\Psiphi}{\Psi_{\varphi}}

\usepackage[left=2.5cm,right=2.5cm,top=2.5cm,bottom=2.5cm]{geometry}

\begin{document}

\begin{frontmatter}

\title{Fleming--Viot Selection of the Yaglom Limit for
Age-Structured Bellman--Harris Processes, with Application to
Livestock Epidemic Surveillance}

\author[lmbp]{Ouerdia Arezki\corref{cor1}}
\ead{ouerdia.arezki@uca.fr}
\author[lmbp]{Paul-Marie Grollemund}
\ead{paul-marie.grollemund@uca.fr}
\author[ul]{Ali Zemouche}
\ead{ali.zemouche@univ-lorraine.fr}

\affiliation[lmbp]{organization={Laboratoire de Math\'ematiques Blaise Pascal
(LMBP), CNRS UMR 6620, Universit\'e Clermont Auvergne},
  addressline={3 place Vasarely}, city={Aubi\`ere}, postcode={63178},
  country={France}}
\affiliation[ul]{organization={Universit\'e de Lorraine, CNRS, CRAN},
  city={Metz}, postcode={F-57000}, country={France}}
\cortext[cor1]{Corresponding author.}

\begin{highlights}
\item First Fleming--Viot particle system for a class of subcritical
Bellman--Harris processes.
\item Particles are age configurations: bounding the population size does not
bound the state.
\item Additive Foster--Lyapunov drift: reproductive value plus exponential age
weight.
\item Explicit Doeblin minorisation and uniform comparison of survival
probabilities.
\item Yaglom theorem in total variation, at an exponential rate in time,
with decay rate the Malthusian parameter.
\item Selection of the Yaglom limit by the particle system, at a
polynomial rate in the number of particles.
\end{highlights}

\begin{abstract}
In this paper, we construct a Fleming--Viot particle system for a class of
subcritical Bellman--Harris processes. We prove that it selects the Yaglom
limit at a polynomial rate in the number of particles. Since lifetimes are
non-exponential, the population size is not Markov, and the analysis must
therefore be carried out on the space of age configurations. In this setting,
the Lyapunov functions used for Galton--Watson processes are no longer
norm-like. Nevertheless, we establish a Yaglom theorem that strengthens the
classical result: the conditional laws converge in total variation at an
exponential rate, with decay rate given by the Malthusian parameter. We also
prove that the drift condition, which links the hazard rate to the offspring
law, is necessary within a natural class of Lyapunov functions, showing that
it is a feature of the measure-valued lift rather than a defect of the
estimates. Finally, we illustrate the estimator through an application to
livestock epidemic surveillance, where the Yaglom limit is the null
distribution of a change-detection test.
\end{abstract}


\end{frontmatter}


\section{Introduction}
\label{sec:intro}

Bellman--Harris branching processes describe age-dependent populations in which
reproduction occurs only at the end of an individual's
lifetime~\cite{BellmanHarris1948,Harris1963,Jagers1975,AthreyaNey1972}. Unlike
continuous-time Galton--Watson processes, the future evolution of such a
population is not determined by its current size alone: it also depends on the
ages of all living individuals. Consequently, the population-size process is,
in general, not Markovian. A subcritical branching population dies out almost
surely, so its long-time behaviour is meaningful only conditionally on
survival. The relevant object is the Yaglom limit~\cite{Yaglom1947}, a
quasi-stationary distribution (QSD) of the killed semigroup, one of possibly
many~\cite{SVJ1966,Maillard2018}; see~\cite{MelVil2012} for a survey.
Estimating it by straightforward simulation becomes impractical, because the
survival probability decays exponentially.

To avoid this difficulty, the Fleming--Viot particle
system~\cite{BHIM1996,BHM2000} keeps the number of live trajectories constant:
when a particle reaches the absorbing state, it is immediately replaced by a copy of
a particle chosen uniformly among the survivors, and the empirical measure of the
$N$ particles then approximates the law conditioned on survival. Consistency is
known for diffusions killed at the boundary of a
domain~\cite{BHM2000,Villemonais2011}, for finite and countable state
spaces~\cite{FerrariMaric2007,AFG2011,CloezThai2016} and for Galton--Watson
processes~\cite{AFGJ2016}; central limit theorems are also
available~\cite{CDGR2020,LelievrePVR2018}. From these results, Champagnat and
Villemonais~\cite{CV2021} derived a tractable criterion from their general
quasi-stationarity framework~\cite{CV2023}: on a Polish space,
and for a bounded killing rate, the minimal QSD is selected as soon as the
unkilled process satisfies three conditions, namely a Foster--Lyapunov drift
whose rate exceeds the killing rate, the small-set property for compact sets,
and a uniform comparison of survival probabilities on compact sets.

However, the Bellman--Harris process does not fit that framework. Its lifetimes
are not exponential, so the population size is no longer Markov, and
Markovianity is recovered only after lifting the dynamics to the space of age
configurations~\cite{Tran2008,Chauvin1986}. This lift has a decisive geometric
consequence: bounding the population size no longer bounds the state. The
sublevel sets of any function of the mass alone contain the entire singleton
stratum, which is closed, unbounded, and has no limit point. Consequently, the
powers of the population size that serve as Lyapunov functions in the
Galton--Watson case are no longer norm-like, and coercivity must be secured in
the mass direction and in the age direction at once. Extinction, on the other
hand, raises no such difficulty: it can occur only through the childless death
of a lone individual, so the killing rate is bounded, and the lifted process
falls under soft killing~\cite{BHM2000}.

In this paper, we remove both obstacles. We construct a Lyapunov function that
is norm-like on the lifted space, coercive in the mass direction and in the age
direction alike. We then establish, on its sublevel sets, the Doeblin
minorisation that a discrete state space would obtain from irreducibility and a
diffusion from Harnack inequalities, neither route being open here. Finally, we
prove that the associated Fleming--Viot system selects the Yaglom limit at a
polynomial rate in the number of particles. Three estimates, each specific to the measure-valued setting, carry the
analysis.
\begin{enumerate}
\item A Foster--Lyapunov drift (Lemma~\ref{lem:lyapunov}): its test function adds two terms of different nature: a power of the reproductive value, which contracts the mass at the Malthusian rate, and an exponential age weight, which
secures coercivity in the age direction. The superposition is additive rather than multiplicative, so the age weight enters linearly and generates no convexity remainder under the generator. The analysis also shows that, at large ages, the unkilled process escapes not at the limiting hazard rate, but at that rate weighted by the probability of a non-zero offspring: a lone individual leaves the singleton stratum only through a birth event.
\item An explicit Doeblin minorisation (Lemma~\ref{lem:smallsets}), whose reference measure is carried by the singleton stratum, the very stratum that obstructed coercivity. From any sublevel set of the Lyapunov function, one drives the population down to a single individual of prescribed age in two phases: extinction of the extra lineages, then resetting of the age through a single birth.
\item A uniform comparison of survival probabilities on compact sets (Lemma~\ref{lem:survival}), obtained by comparing lineages age by age in both directions, without recourse to renewal asymptotics.
\end{enumerate}

From these estimates, our main result follows: the Fleming--Viot system selects the Yaglom limit (Theorem~\ref{thm:convergence}) and does so quantitatively. For every bounded test function, the mean discrepancy between the stationary empirical measure and the Yaglom limit is bounded by a negative power of the number of particles, so we obtain convergence in probability; the exponent, expressed in terms of the killing rate and the contraction rate of the conditioned semigroup, is discussed in Remark~\ref{rem:rate}. Along the way, the argument also re-derives the Yaglom limit itself in total variation, at an exponential rate, and identifies its decay rate as the Malthusian parameter (Remark~\ref{rem:yaglom}). A consistency result accompanies it (Theorem~\ref{thm:consistency}), independent of~\cite{CV2021}: every deterministic limit point of the stationary empirical measure is a QSD, whose
decay rate equals the mean killing intensity under that distribution.

To our knowledge, this is the first Fleming--Viot approximation of a Yaglom
limit for an age-dependent branching process. The lift is forced by the loss of
the Markov property, and the drift condition it brings with it is no technical
convenience either: we prove it necessary, and not merely sufficient
(Proposition~\ref{prop:sharp}). It couples the lifetime law to the offspring law
and restricts the mean offspring number to $m>\tfrac12$, a restriction the
Markovian criterion of~\cite{CV2021} does not carry; Remark~\ref{rem:markov}
exhibits a continuous-time Galton--Watson process that it excludes and that
criterion covers. It is thus the price of the measure-valued lift: a feature of
the geometry of $\Mp(\R_+)$, not a defect of the estimates.
The paper is organised as follows. Section~\ref{sec:prelim} constructs the
lifted process, establishes its extended generator, states the assumptions, and
defines the particle system. Section~\ref{sec:main-results} states the three
estimates, proves the necessity of the drift condition, and deduces the two
theorems; the proofs are deferred to the appendices.
Section~\ref{sec:app} illustrates the estimator on a surveillance problem, in
which the Yaglom limit is the null distribution of a change-detection test.
Section~\ref{sec:conclusion} concludes.
\section{Model, assumptions, and the particle system}
\label{sec:prelim}
Let $G$ be a lifetime law on $[0,\infty)$, absolutely continuous with continuous
density $g$ and $G(0)=0$, and let $h(s)=\sum_{n\ge0}p_ns^n$ be the offspring
generating function, with offspring variable $\xi\sim(p_n)_{n\ge0}$, mean
$m:=h'(1)<\infty$ and $p_0:=\Prob(\xi=0)$. In the Bellman--Harris process with
parameters $(G,h)$, each individual lives a lifetime $T\sim G$ and, at death, is
replaced by $\xi$ offspring born at age zero, all lifetimes and offspring
numbers being independent. Write $\bar G:=1-G$, let $\mu:=g/\bar G$ denote the
hazard rate, so that $\bar G(a)=\exp(-\int_0^a\mu)$, and set
\[
\mu^*:=\sup_{a\ge0}\mu(a),\; \mu_\infty:=\liminf_{a\to\infty}\mu(a),
\;\text{so that } \mu_\infty\le\mu^*.
\]
The death intensity $\sum_{i\in\mathcal V_t}\mu(a_i(t))$, where $\mathcal V_t$
is the set of individuals alive at time $t$ and $a_i(t)$ the age of
$i\in\mathcal V_t$, depends on the ages of the living individuals, not on their
number alone, so the population size is not Markov.
Following Tran~\cite{Tran2008}, Markovianity is restored on the space
$\Mp(\R_+)$ of finite point measures, endowed with the narrow topology, and we
set
\[
\eta_t:=\sum_{i\in\mathcal V_t}\delta_{a_i(t)},\;
Z_t:=\langle\eta_t,\ind\rangle,\;
\tau:=\inf\{t>0:\eta_t=\mathbf 0\},\;
E:=\Mp(\R_+)\setminus\{\mathbf 0\}.
\]

The process $(\eta_t)_{t\ge0}$ is strong Markov~\cite[Theorem~2.5]{Tran2008},
and its extended generator in the sense of Meyn and Tweedie~\cite{MT1993} reads
as follows. A measurable $F:E\to\R$ belongs to that extended domain when
there is a measurable $\mathcal LF:E\to\R$ such that
$F(\eta_t)-F(\eta_0)-\int_0^t\mathcal LF(\eta_s)\,ds$ is a local martingale
under $\Prob_\eta$ for every $\eta\in E$; the localisation recorded
after~\eqref{eq:sup-Z} is what places the unbounded $V$ of~\eqref{eq:lyapunov}
in this domain. For $F(\eta)=\Phi(\langle\eta,f\rangle)$ with
$\Phi\in C^1(\R)$ and $f\in C^1(\R_+)$, it reads
\begin{align}
\label{eq:generator}
\mathcal LF(\eta)
&=\Phi'\bigl(\langle\eta,f\rangle\bigr)\langle\eta,f'\rangle \notag\\
&\quad+\int_{\R_+}\!\!\mu(a)\sum_{n\ge0}p_n
\Bigl[\Phi\bigl(\langle\eta,f\rangle-f(a)+nf(0)\bigr)
-\Phi\bigl(\langle\eta,f\rangle\bigr)\Bigr]\eta(da),
\end{align}
individuals aging at unit speed between jumps and an individual of age $a$ being
replaced at a rate of $\mu(a)$ through $\eta\mapsto\eta-\delta_a+\xi\delta_0$.
Throughout, $f$ denotes a function on ages, $f:\R_+\to\R$, and $F$ a
functional on configurations, $F:E\to\R$; functionals on the product space $E^N$
of Section~\ref{sec:prelim} below are denoted by $\Psi$. Taking
$\Phi=\mathrm{id}$ linearises this to
$\mathcal L\langle\eta,f\rangle=\langle\eta,\mathcal Af\rangle$, where
\begin{equation}
\label{eq:generator-linear}
(\mathcal Af)(a):=f'(a)+\mu(a)\bigl(m\,f(0)-f(a)\bigr).
\end{equation}

Two elementary moment bounds for $\eta$ are used repeatedly; the corresponding
bounds for $X$, which differ, are established after~\eqref{eq:LX}. Taking
$f=\ind$ in~\eqref{eq:generator-linear} gives $\mathcal A\ind=(m-1)\mu$, hence
\begin{equation}
\label{eq:LZ}
\mathcal LZ(\eta)=\sum_{i}\mu(a_i)\,(m-1)\ \le\ 0,\; \eta\in E,
\end{equation}
by {\rm(H2)}: under $\eta$ the mass is a supermartingale, and it is exactly this
monotonicity that~\eqref{eq:LXZ} will show to fail under $X$. To exploit it we
first rule out explosion. Between jumps the ages grow at unit speed and the
total jump rate at $\eta$ is $\sum_i\mu(a_i)\le\mu^*Z$, so jump times can
accumulate only if $Z$ reaches infinity in finite time. Let
$H_L:=\inf\{t\ge0:Z_t\ge L\}$ for an integer $L>Z_0$. On $[0,H_L)$ the mass is
at most $L-1$ and the jump rate at most $\mu^*(L-1)$, so only finitely many
jumps occur there, and $Z_{H_L}\le L-2+\xi$; since $\E[\xi]=m<\infty$, the
random variable $Z_{t\wedge H_L}$ is integrable, and Dynkin's formula on
$[0,t\wedge H_L]$ gives $\E_\eta[Z_{t\wedge H_L}]\le Z_0$ by~\eqref{eq:LZ}. As
$Z_{H_L}\ge L$ on $\{H_L\le t\}$, Markov's inequality gives
$\Prob_\eta(H_L\le t)\le Z_0/L\to0$: the process $\eta$ is non-explosive, and
Fatou's lemma yields
\begin{equation}
\label{eq:mean-Z}
\E_\eta[Z_t]\ \le\ Z_0,\; t\ge0.
\end{equation}
The mass increases only at births, whose expected rate is
$m\sum_i\mu(a_i)\le m\mu^*Z$, so $\sup_{s\le t}Z_s\le Z_0+B_t$ with $B_t$ the
number of births on $[0,t]$; the same localisation and~\eqref{eq:mean-Z} give
\begin{equation}
\label{eq:sup-Z}
\E_\eta\Bigl[\sup_{s\le t}Z_s\Bigr]\ \le\ Z_0\bigl(1+m\mu^*t\bigr),\; t\ge0.
\end{equation}

Stratifying $\Mp(\R_+)=\bigsqcup_{n\ge0}\{Z=n\}$ with
$\{Z=n\}\cong\R_+^{\,n}/\mathfrak S_n$, one sees that a subset is relatively
compact if and only if masses and ages are bounded on it; the strata are
open and closed, since $\eta\mapsto\langle\eta,\ind\rangle$ is narrowly
continuous and integer-valued. Bounding the population size therefore
does not bound the state, the singleton ray $\{\delta_a:a\ge0\}$ has mass one
throughout and no limit point. Consequently the sublevel sets of
$\eta\mapsto\langle\eta,\ind\rangle^\alpha$, or of any function of the mass
alone, fail to be relatively compact, and the Lyapunov functions used in the
Galton--Watson case are not norm-like here. Supplying a substitute, which
confines masses and ages simultaneously, is the object of
Lemma~\ref{lem:lyapunov}. A branching event takes the mass from $Z$ to $Z-1+n$, which vanishes only if
$Z=1$ and $n=0$, extinction occurs only from a singleton, and only through a
childless death. Reading the rates off~\eqref{eq:generator}, the killing rate is
\begin{equation}
\label{eq:kappa}
\kappa(\eta)=p_0\,\mu(a)\,\ind_{\{\eta=\delta_a\}},
\; \|\kappa\|_\infty=p_0\,\mu^*<\infty \;\text{under {\rm(H4)}},
\end{equation}
where $\|\phi\|_\infty:=\sup_{\eta\in E}|\phi(\eta)|$ denotes the supremum
norm on $E$, used throughout for functions on $E$ as well as, in
Section~\ref{sec:main-results}, for bounded measurable test functions; the
second equality holds because $\kappa$ vanishes off the singleton stratum, so
that the supremum over $E$ reduces to $p_0\sup_{a\ge0}\mu(a)$,
and $\kappa\not\equiv0$, since $p_0\ge1-m>0$ under {\rm(H2)}. Absorption at
$\mathbf 0$, a hard boundary for the population-size process, is thus a
bounded killing rate on the lifted space, and the model enters the
soft-killing framework of~\cite{CV2021}. Let $X$ be the process on $E$ obtained
by deleting the transition $\delta_a\mapsto\mathbf 0$, when a childless death
fires at a singleton, nothing happens, and the individual keeps ageing. This is
the only possible choice, since the killing clock of~\cite{CV2021} is
independent of $X$, any transition performed by $X$ at rate $\kappa$ would be
visible before killing and would alter the law of $\eta$. Then, $\eta$ is $X$
killed at rate $\kappa$, and for every $F$ in the extended domain,
\begin{equation}
\label{eq:LX}
\mathcal L_XF(\eta)=\mathcal LF(\eta)+\kappa(\eta)\bigl(F(\eta)-F(\mathbf 0)\bigr),
\; \eta\in E.
\end{equation}

The criteria of~\cite{CV2021} bear on $\mathcal L_X$,
and \eqref{eq:LX} shows that a drift inequality for $\mathcal L$ does not
transfer to it: under $X$, a lone individual of age $a$ leaves the large-age
region only through a birth event with $n\ge1$, hence at rate $(1-p_0)\mu(a)$
and not $\mu(a)$. The age-confinement rate genuinely available is
$(1-p_0)\mu_\infty$, and it is this quantity, not $\mu_\infty$, that the drift
of Lemma~\ref{lem:lyapunov} must exceed. We record at this point the moment bounds for $X$, which differ from those for
$\eta$ and for which we write $\E^X_\eta$. Deleting the childless death at a
singleton removes the only event that decreases the mass there:
applying~\eqref{eq:LX} to
$Z=\langle\cdot,\ind\rangle$, for which $\mathcal LZ(\eta)=\sum_i\mu(a_i)(m-1)$
by~\eqref{eq:LZ} and $Z(\mathbf 0)=0$,
\begin{equation}
\label{eq:LXZ}
\mathcal L_XZ(\eta)=\sum_{i}\mu(a_i)(m-1)+\kappa(\eta)Z(\eta),
\text{ so that }
\mathcal L_XZ(\delta_a)=\mu(a)\bigl(m-1+p_0\bigr)\ \ge\ 0 ,
\end{equation}
the inequality being $p_0\ge1-m$, with equality if and only if $\xi\in\{0,1\}$
almost surely. The mass drifts upwards along the singleton stratum under $X$, so
$\E^X_\eta[Z_t]$ need no longer be non-increasing, and the monotonicity
of~\eqref{eq:LZ} is unavailable. For the same reason, $X$ is not a
branching process: an individual is replaced at rate $\mu(a)$ when it is not
alone and at rate $(1-p_0)\mu(a)$ when it is. The branching property is
therefore invoked below for $\eta$ only.

What survives is an exponential bound, which is all that is needed. Since
$\kappa\ge0$ and $m<1$, \eqref{eq:LXZ} gives
$\mathcal L_XZ\le\|\kappa\|_\infty Z$ on $E$. With $H_L$ as above, Dynkin's
formula applied to
$(s,\eta)\mapsto e^{-\|\kappa\|_\infty s}Z(\eta)$ on $[0,t\wedge H_L]$, whose
space--time generator
$e^{-\|\kappa\|_\infty s}\bigl(\mathcal L_XZ-\|\kappa\|_\infty Z\bigr)$ is
non-positive, gives
$\E_\eta\bigl[e^{-\|\kappa\|_\infty(t\wedge H_L)}Z_{t\wedge H_L}\bigr]\le Z_0$,
hence $\E_\eta[Z_{t\wedge H_L}]\le Z_0e^{\|\kappa\|_\infty t}$.

As
$Z_{H_L}\ge L$ on $\{H_L\le t\}$, Markov's inequality gives
$\Prob_\eta(H_L\le t)\le Z_0e^{\|\kappa\|_\infty t}/L\to0$ as $L\to\infty$: the
process $X$ is non-explosive. Fatou's lemma then yields
\begin{equation}
\label{eq:mean-Z-X}
\E^X_\eta[Z_t]\ \le\ Z_0\,e^{\|\kappa\|_\infty t},\; t\ge0.
\end{equation}
Off the singletons a branching event occurs at rate $\mu(a_i)$ per individual
with offspring mean $m$; on them it occurs at the reduced rate $(1-p_0)\mu(a)$
with the raised mean $\E[\xi\mid\xi\ge1]=m/(1-p_0)$. The two compensate exactly,
so the expected birth rate is $m\sum_i\mu(a_i)\le m\mu^*Z$ in both cases. As $Z$
increases only at births, \eqref{eq:mean-Z-X} and
$\|\kappa\|_\infty=p_0\mu^*$ of~\eqref{eq:kappa} give
\begin{equation}
\label{eq:sup-Z-X}
\E^X_\eta\Bigl[\sup_{s\le t}Z_s\Bigr]
\ \le\ Z_0\Bigl(1+\frac{m}{p_0}\bigl(e^{\|\kappa\|_\infty t}-1\bigr)\Bigr),
\; t\ge0.
\end{equation}

Identity~\eqref{eq:generator-linear} extends to $f$ of at most exponential
growth, the regime of the age weight $\rho(a)=e^{\beta a}$ used below: ages grow
at unit speed, so $\rho(a_i(t))\le e^{\beta t}\rho(a_i(0))$ for the individuals
alive at time $0$ and $\rho(a_i(t))\le e^{\beta t}$ for those born later, whence
$\langle\eta_t,\rho\rangle
\le e^{\beta t}\bigl(\langle\eta_0,\rho\rangle+\sup_{s\le t}Z_s\bigr)$
is integrable, by~\eqref{eq:sup-Z} under $\eta$ and by~\eqref{eq:sup-Z-X} under
$X$; localisation at the exit times of the sublevel sets of the Lyapunov
function, together with the change-of-variable formula of
Tran~\cite[Lemma~2.6]{Tran2008} and Lemma~\ref{lem:malthus}(i), then applies. Since extinction is almost sure, the long-time behaviour is described
conditionally on survival, through the sub-Markov killed semigroup
$P_t^{\kappa}\varphi(\eta):=\E_\eta[\varphi(\eta_t)\ind_{\{\tau>t\}}]$,
$\eta\in E$, whose generator we write $\mathcal L^{\kappa}$ and which is
dominated by the semigroup $P_t\varphi(\eta):=\E_\eta[\varphi(\eta_t)]$ of the
unkilled process $X$. We write $P_t(\eta,\cdot)$ and $P_t^{\kappa}(\eta,\cdot)$
for the associated kernels, so that $P_t\varphi$ and $\nu P_t^{\kappa}$ denote
their action on bounded measurable functions and on measures respectively. The
long-time behaviour is governed by the Malthusian parameter $\lambda_0$, the
unique root of
\begin{equation}
\label{eq:malthus}
m\int_0^\infty e^{\lambda t}\,dG(t)=1.
\end{equation}
A probability measure $\nu$ on $E$ is quasi-stationary with decay rate
$\lambda>0$, a scalar, if $\nu P_t^{\kappa}=e^{-\lambda t}\nu$ for all $t\ge0$;
the letter $\lambda$ is otherwise used only as the unknown
of~\eqref{eq:malthus}. Such measures need not be unique: for Galton--Watson
processes they form a one-parameter family~\cite{SVJ1966,Maillard2018}, and the
same multiplicity is to be expected here. It is this multiplicity that makes the
selection question non-trivial.
\begin{remark}
\label{rem:minimal}
We call $\nuY$ the Yaglom limit, in the classical sense of the weak
limit of the conditional laws $\Prob_{\delta_0}(\eta_t\in\cdot\mid\tau>t)$, and
this is what Theorem~\ref{thm:convergence}(ii) establishes. Two senses of
minimality must be distinguished. In the sense of~\cite{CV2021,CV2023}, where
the minimal QSD is by definition the one whose domain of attraction contains
every Dirac mass, $\nuY$ is minimal: the function $V$
of~\eqref{eq:lyapunov} is finite at every $\eta\in E$, so
Theorem~\ref{thm:convergence}(i) applies to $\pi_0=\delta_\eta$ for every
$\eta\in E$. What we do not claim is minimality in the ordering sense: our
results give uniqueness within the class $\{\nu:\nu(V)<\infty\}$, not the
inequality $\lambda_0\ge\lambda$ for every other quasi-stationary decay rate
$\lambda$, which we leave open.
\end{remark}

We now list the conditions under which the results of this paper are
established. They bear on the lifetime law $G$ and on the offspring law $h$.
\begin{itemize}\itemsep2pt
\item[(H1)] $G$ is absolutely continuous with continuous density $g$, $G(0)=0$,
$\supp(G)=[0,\infty)$, and $\inf_{[\ell,L]}g>0$ for all $0<\ell<L<\infty$.
\item[(H2)] $m<1$.
\item[(H3)] $p_1=\Prob(\xi=1)>0$.
\item[(H4)] $0<\mu_\infty$ and $\mu^*<\infty$.
\item[(H5)] $p_0\,\mu^*<(1-p_0)\,\mu_\infty$.
\item[(H6)] $\sum_{n\ge0}n^\alpha p_n<\infty$ for some $\alpha>1$ with
$\alpha\,\lambda_0>p_0\,\mu^*$.
\end{itemize}
One such scalar $\alpha>1$ is fixed once and for all, before the particle system
is introduced; every constant below may depend on it, and none depends on $N$.
Since $\mu^*\ge\mu_\infty$, {\rm(H5)} forces $p_0<\tfrac12$, hence $m>\tfrac12$.
In the notation of~\eqref{eq:kappa}, {\rm(H5)} reads
$\|\kappa\|_\infty<(1-p_0)\mu_\infty$: the killing rate is beaten by the age-confinement rate. This is the form in which it will be used throughout, and
{\rm(H6)} likewise reads $\alpha\lambda_0>\|\kappa\|_\infty$. We record for later use that $p_0\ge 1-m>0$ under {\rm(H2)}, since $\E[\xi]\ge\Prob(\xi\ge1)=1-p_0$; consequently $(1-p_0)\mu_\infty\le(1-p_0)\mu^*<\mu^*$, an inequality used in the proof of Lemma~\ref{lem:malthus}(ii).

The Malthusian parameter and the reproductive value, which govern the two terms of the Lyapunov function of Section~\ref{sec:main-results}, are classical objects~\cite[Chapter~IV]{AthreyaNey1972}; we record them in the form used
below.

\begin{lemma}
\label{lem:malthus}
Assume {\rm(H1)}, {\rm(H2)}, {\rm(H4)}, {\rm(H5)}. Then \emph{(i)}
$\int_0^\infty e^{\theta a}\,dG(a)<\infty$ for every scalar $\theta<\mu_\infty$;
\emph{(ii)} equation~\eqref{eq:malthus} has a unique root $\lambda_0$, a scalar
with $0<\lambda_0<(1-p_0)\mu_\infty$; \emph{(iii)} the reproductive value
\begin{equation}
\label{eq:reproductive-value}
v(a):=m\,\E\bigl[e^{\lambda_0(T-a)}\mid T>a\bigr]
=\frac{m\,e^{-\lambda_0a}}{\bar G(a)}\int_a^\infty e^{\lambda_0s}\,dG(s)
\end{equation}
is a function $v:\R_+\to\R_+$ belonging to $C_b^1(\R_+)$, and it satisfies
$v(0)=1$, $m\le v\le v_{\max}<\infty$ for a finite scalar $v_{\max}$, and
$\mathcal Av=-\lambda_0v$, so that
$\mathcal L\langle\eta,v\rangle=-\lambda_0\langle\eta,v\rangle$ for every
$\eta\in E$.
\end{lemma}

\begin{proof}
See~\ref{app:malthus}.
\end{proof}

The upper bound $\lambda_0<(1-p_0)\mu_\infty$ in (ii) is a consequence
of {\rm(H5)}, not a separate Malthusian hypothesis; it places $\lambda_0$
strictly below $\mu_\infty$, which is what makes the reproductive value bounded
in~(iii). Two families of lifetime laws show that {\rm(H5)} is neither vacuous nor
automatic. For a Gamma law $\Gamma(k,\vartheta)$ with shape $k\ge1$ and scale
$\vartheta>0$ the hazard rate increases to $1/\vartheta$, so
$\mu^*=\mu_\infty=1/\vartheta$ and {\rm(H5)} reduces to $p_0<\tfrac12$; here
$\E[e^{\lambda T}]=(1-\lambda\vartheta)^{-k}$ for $\lambda<1/\vartheta$, whence
$\lambda_0=\vartheta^{-1}(1-m^{1/k})$ explicitly. For a hyperexponential law
$\bar G(a)=\sum_{i=1}^rq_ie^{-\mu_ia}$ with $q_i>0$, $\sum_iq_i=1$ and
$0<\mu_1<\cdots<\mu_r$, the hazard rate decreases from $\mu^*=\sum_iq_i\mu_i$ to
$\mu_\infty=\mu_1$, and {\rm(H5)} becomes
$p_0<\mu_1/(\mu_1+\sum_iq_i\mu_i)$, which fails as soon as the population is
both strongly dispersed in its lifetimes and strongly prone to childless death.
Assumption {\rm(H4)} excludes the Weibull laws of shape $k\ne1$: for $k>1$ the
hazard rate increases without bound, so that $\mu^*=\infty$, while for $k<1$ it
is unbounded near the origin and decreases to zero, so that $\mu^*=\infty$ and
$\mu_\infty=0$; in either case {\rm(H4)} fails. It also excludes, through
$\mu_\infty>0$, the log-normal law, which has no exponential moment, so
that~\eqref{eq:malthus} has no root at all.

\medskip
Fix an integer $N\ge2$, the particles
$\mathbf X_t^N=(X_t^{N,1},\dots,X_t^{N,N})\in E^N$ evolve as independent copies
of the unkilled process $X$ and, at rate $\kappa(X_t^{N,i})$, particle $i$ is
replaced by a copy of one drawn uniformly among the $N-1$ others. Each particle
is a full age configuration, so the system lives on $E^N$, a product of copies
of a measure-valued space rather than of a discrete or Euclidean one. Its
extended generator is
\begin{equation}
\label{eq:gen-N}
\mathcal L^{(N)}\Psi(\mathbf x)
=\sum_{i=1}^N\mathcal L_X^{(i)}\Psi(\mathbf x)
+\sum_{i=1}^N\kappa(x^i)\,\frac1{N-1}
\sum_{\substack{j=1\\ j\ne i}}^{N}
\bigl[\Psi(\mathbf x^{\,i\leftarrow j})-\Psi(\mathbf x)\bigr],
\; \mathbf x\in E^N,
\end{equation}
for every functional $\Psi:E^N\to\R$ in the extended domain of
$\mathcal L^{(N)}$, with $\mathcal L_X^{(i)}$ acting on the $i$-th coordinate
through~\eqref{eq:LX} and $\mathbf x^{\,i\leftarrow j}$ denoting $\mathbf x$
with $x^i$ replaced by a copy of $x^j$; the free motion is that of $X$, the
transition of $\eta$ to $\mathbf 0$ being precisely the one that the rebirth
replaces. Since $\|\kappa\|_\infty<\infty$, the total rebirth rate is at most
$N\|\kappa\|_\infty$, so the rebirth times do not accumulate:
$(\mathbf X_t^N)_{t\ge0}$ is well defined, non-explosive and strong Markov on
$E^N$, and remains exchangeable if its initial law is. The object of interest is
the empirical measure
\begin{equation}
\label{eq:empirical}
m_t^N:=\frac1N\sum_{i=1}^N\delta_{X_t^{N,i}}\ \in\ \mathcal P(E),
\end{equation}
a random probability measure on the space of age configurations; we write
$\mathcal X^N$ for its limit in distribution as $t\to\infty$ at fixed $N$, whose
existence is part of Theorem~\ref{thm:consistency}. Three levels are in play and
should not be conflated: $\eta\in E$ is a point measure on $\R_+$; $m_t^N$ and
$\mathcal X^N$ are random elements of $\mathcal P(E)$, that is, random measures
on a space of measures; and the law of $\mathcal X^N$ is an element of
$\mathcal P(\mathcal P(E))$, which is the space in which the convergence just
stated, and the convergence to $\delta_{\nu^*}$ of
Theorem~\ref{thm:consistency}, both take place. Each of $E$, $\mathcal P(E)$ and
$\mathcal P(\mathcal P(E))$ is Polish, the last two under weak convergence. The
random measure $\mathcal X^N$ remains random at every fixed $N$; a deterministic
element of $\mathcal P(E)$ is recovered only in the further limit $N\to\infty$
(Theorem~\ref{thm:convergence}(iii)).

\section{Main results}
\label{sec:main-results}

The criteria of~\cite{CV2021} require three things of the unkilled
process $X$ on the Polish space $E$: a Foster--Lyapunov function $V\ge1$ with
relatively compact sublevel sets and drift rate exceeding $\|\kappa\|_\infty$;
the small-set property for compact sets; and a uniform comparison of survival
probabilities on compact sets. None of the three is automatic on $\Mp(\R_+)$,
and we establish them in turn. The first is a Foster--Lyapunov drift. Its test function superposes a power of the reproductive
value, which contracts the mass at the Malthusian rate, and an exponential
weight in the age, which confines the ages; the superposition is additive, so
that the age term enters linearly and produces no convexity remainder. The
following lemma states it, together with the description of the sublevel sets
that makes $V$ norm-like on the lifted space.

\begin{lemma}
\label{lem:lyapunov}
Assume {\rm(H1)}--{\rm(H6)}, let $\lambda_0$ and $v$ be as in
Lemma~\ref{lem:malthus} and $\alpha$ as in {\rm(H6)}, and fix scalars
$\beta\in\bigl(0,\ (1-p_0)\mu_\infty-\|\kappa\|_\infty\bigr)$,
$\varepsilon'\in\bigl(0,\ (1-p_0)\mu_\infty-\|\kappa\|_\infty-\beta\bigr)$ and
$\theta\in\bigl(0,\ \alpha\lambda_0-\|\kappa\|_\infty\bigr)$,
these ranges being non-empty by {\rm(H5)} and {\rm(H6)}. Set
$\rho:\R_+\to[1,\infty)$, $\rho(a):=e^{\beta a}$, and, for any scalar
$\varepsilon\in(0,1]$, let $V:\Mp(\R_+)\to[1,\infty)$ be given by
\begin{equation}
\label{eq:lyapunov}
V(\eta):=1+\langle\eta,v\rangle^{\alpha}+\varepsilon\,\langle\eta,\rho\rangle,
\; V(\mathbf 0):=1.
\end{equation}
Then there is a finite constant $C$, depending only on $(G,h)$ and on the
parameters $(\alpha,\theta,\varepsilon,\varepsilon',\beta)$ just fixed, and in
particular neither on $\eta$ nor on $N$, such that
\begin{equation}
\label{eq:lyap-global}
\mathcal L_XV(\eta)\ \le\ -\lambda_1V(\eta)+C,\; \eta\in E,
\end{equation}
where $\lambda_1:=\min\bigl(\alpha\lambda_0-\theta,\
(1-p_0)\mu_\infty-\varepsilon'-\beta\bigr)$ satisfies
$\lambda_1>\|\kappa\|_\infty$.
Moreover $V$ is bounded on compact sets, and for every scalar $L\ge1$ the
sublevel set $K_L:=\{V\le L\}$ is contained in the set of configurations
carrying at most $R_L:=\lfloor L^{1/\alpha}/m\rfloor$ individuals, an integer,
all of age at most $A_L:=\beta^{-1}\log(L/\varepsilon)$, a non-negative scalar
since $\varepsilon\le1\le L$; in particular $K_L$ is relatively compact in the
narrow topology, and every compact subset of $E$ is contained in some $K_L$.
\end{lemma}

\begin{proof}
See~\ref{app:lyapunov}.
\end{proof}

The rate $(1-p_0)\mu_\infty$ appearing in $\lambda_1$ is not an artefact of this
particular choice of $V$. On the lifted space it is the largest drift rate any
Lyapunov function can achieve, and the requirement
$\lambda_1>\|\kappa\|_\infty$ of~\cite{CV2021} then reduces exactly
to {\rm(H5)}. This is the content of the next proposition, which shows the
drift condition to be necessary and not merely sufficient.

\begin{proposition}
\label{prop:sharp}
Assume {\rm(H2)} and {\rm(H4)}. Let $V:\Mp(\R_+)\to[1,\infty)$ belong to the
extended domain of $\mathcal L_X$, have relatively compact sublevel sets, and be
such that the scalar $B_V:=\sum_{n\ge1}p_n\,V(n\delta_0)$ is finite and the
function $W:\R_+\to[1,\infty)$, $W(a):=V(\delta_a)$, is $C^1$ and non-decreasing
for large ages. If $\mathcal L_XV\le-\lambda_1V+C$ on $E$, for some finite constant $C$ and some scalar $\lambda_1>\|\kappa\|_\infty$, then {\rm(H5)} holds.
Equivalently, if {\rm(H5)} fails then no such $V$ exists, and the criteria
of~\cite{CV2021} are unavailable on $\Mp(\R_+)$.
\end{proposition}

\begin{proof}
Under $X$ a lone individual of age $a$ ages at unit speed and, at rate $\mu(a)$,
is replaced by $n$ offspring; the childless event $n=0$ has been deleted, so it
leaves the singleton stratum only through a birth, and $\sum_{n\ge1}p_n=1-p_0$.
Hence $\mathcal L_XV(\delta_a)=W'(a)-(1-p_0)\mu(a)W(a)+\mu(a)B_V$, and the drift
inequality, evaluated at $\eta=\delta_a$, reads
\begin{equation*}  
[\lambda_1-(1-p_0)\mu(a)]W(a)\le C-W'(a)-\mu(a)B_V\le C\end{equation*} 
for every $a$ large enough, the last step using $W'\ge0$ at large ages and $B_V\ge0$. The sublevel set $\{V\le L\}$ is relatively compact, hence carries bounded ages,
so it meets the singleton stratum in a bounded range of ages; therefore
$W(a)\to\infty$. Suppose $\lambda_1>(1-p_0)\mu_\infty$. Choosing
$a_k\to\infty$ with $\mu(a_k)\to\mu_\infty$ makes the bracket above bounded
below by $\varsigma>0$ for $k$ large, whence
$\varsigma\,W(a_k)\le C$, contradicting $W(a_k)\to\infty$. Consequently
$\lambda_1\le(1-p_0)\mu_\infty$, and together with
$\lambda_1>\|\kappa\|_\infty=p_0\mu^*$ this gives
$p_0\mu^*<(1-p_0)\mu_\infty$.
\end{proof}

The class is not empty, the function~\eqref{eq:lyapunov} lies in it, since
$v(0)=\rho(0)=1$ gives $V(n\delta_0)=1+n^\alpha+\varepsilon n$, whence
$B_V\le1+\E[\xi^\alpha]+\varepsilon m<\infty$ by {\rm(H6)}, while
$W(a)=1+v(a)^\alpha+\varepsilon e^{\beta a}$ is $C^1$ and increasing for large
$a$, the exponential term dominating a bounded derivative.The monotonicity requirement on $W$ is what discards the term $-W'(a)$ above.

We turn to the small-set property. On a discrete state space it follows from
irreducibility, and for diffusions from Harnack inequalities; on $\Mp(\R_+)$
neither route is open, and the minorisation is built by hand. We drive the
population down to a single individual of prescribed age, so that the reference
measure is carried by the singleton stratum. The next lemma makes this explicit.
\begin{lemma}
\label{lem:smallsets}
Assume {\rm(H1)}--{\rm(H6)}. For every scalar $L\ge1$ there exist scalars
$t_L>0$ and $\alpha_L\in(0,1]$ and a probability measure $\nu_L$ on $E$, all
depending only on $(G,h)$ and $L$, and, unlike the constants of
Lemma~\ref{lem:lyapunov}, on $L$ genuinely, such that
\begin{equation}
\label{eq:doeblin}
P^{\kappa}_{t_L}(\eta,\cdot)\ \ge\ \alpha_L\,\nu_L(\cdot),\; \eta\in K_L.
\end{equation}
The measure $\nu_L$ is the law of $\delta_{\mathcal U}$, where $\mathcal U$ is a
random variable uniform on an interval $[c_1,c_2]\subset(0,\infty)$, the letter
$C$ being reserved for scalars, and is therefore carried by the singletons.
Since $P_{t_L}\ge P^{\kappa}_{t_L}$, every compact subset of $E$ is a small set
for $X$.
\end{lemma}

\begin{proof}
See~\ref{app:smallsets}.
\end{proof}

The third estimate compares survival probabilities on a compact set. A lineage
started at age $a$ can be turned into a lineage started at any other age $a'$,
at a cost bounded below uniformly in $t$: by waiting, if $a<a'$, and by dying
with exactly one offspring and then waiting, if $a>a'$. The second mechanism is
where {\rm(H3)} enters. No renewal asymptotics are needed. The following lemma
states the resulting bound.

\begin{lemma}
\label{lem:survival}
Assume {\rm(H1)}--{\rm(H6)}. For every compact $K\subset E$,
\[
\inf_{t\ge0}\
\frac{\inf_{\eta\in K}\Prob_\eta(\tau>t)}{\sup_{\eta\in K}\Prob_\eta(\tau>t)}
\ \ge\ \frac{\sigma_A}{R}\ >\ 0 ,
\]
where the scalars $R$ and $A$ bound respectively the mass and the ages on $K$,
as in Lemma~\ref{lem:lyapunov}, and $\sigma_A>0$ is the explicit scalar
constructed in~\ref{app:survival}.
\end{lemma}

\begin{proof}
See~\ref{app:survival}.
\end{proof}

We now identify the limit of the empirical measure. The test functions are the
exponential functionals of $\eta$, which are bounded, continuous, vanish at the
cemetery, and, unlike a general class $\{\Phi(\langle\cdot,f\rangle)\}$, are
stable under the killed semigroup. Say that $\varphi:E\to\R$ is a test
functional if $\varphi=\varphi_\psi$ for some $\psi\in C_b^1(\R_+)$ with
$\psi\le0$, where
$\varphi_\psi(\eta):=1-e^{\langle\eta,\psi\rangle}=1-\prod_ie^{\psi(a_i)}$ for
$\eta\in\Mp(\R_+)$.
Such a $\varphi_\psi$ is of the form $\Phi(\langle\eta,\psi\rangle)$ with
$\Phi(x)=1-e^{x}$, which is $C^1$ on $(-\infty,0]$ with $0\le\Phi\le1$,
$|\Phi'|\le1$ and $\Phi(0)=0$; test functionals are therefore bounded,
continuous, and in the extended domain of $\mathcal L$. They vanish at the
cemetery, $\varphi_\psi(\mathbf 0)=0$, so that
$\mathcal L^{\kappa}\varphi_\psi=\mathcal L\varphi_\psi$ and
$P_t^{\kappa}\varphi_\psi(\eta)=\E_\eta[\varphi_\psi(\eta_t)]$. They are
measure-determining on $E$, since $\{\nu(\varphi_\psi)\}_\psi$ determines the
Laplace functional
$\psi\mapsto\int_Ee^{\langle\eta,\psi\rangle}\,\nu(d\eta)$. Finally,
by~\eqref{eq:generator}, with the finite scalar
$c_\psi:=\|\psi'\|_\infty+2\mu^*$,
\begin{equation}
\label{eq:Lphi-growth}
|\mathcal L\varphi_\psi(\eta)|\ \le\ c_\psi\,\langle\eta,\ind\rangle
\ \le\ \frac{c_\psi}{m}\,V(\eta)^{1/\alpha}\ =\ o\bigl(V(\eta)\bigr),
\end{equation}
so $\mathcal L\varphi_\psi$ is continuous and of order $o(V)$, a fact used twice
below. The next lemma shows that this class is preserved by the killed
semigroup, which is what makes it usable in the arguments that follow.

\begin{lemma}
\label{lem:stability}
Assume {\rm(H1)}, {\rm(H2)}, {\rm(H4)}. For every $\psi\in C_b^1(\R_+)$ with
$\psi\le0$ and every $t\ge0$ there is $\psi_t\in C_b^1(\R_+)$ with
$\psi_t\le0$ such that $P_t^{\kappa}\varphi_\psi=\varphi_{\psi_t}$.
\end{lemma}

\begin{proof}
See~\ref{app:stability}.
\end{proof}

The first theorem holds for every integer $N\ge2$ and does not use the survival
comparison: it says that the particle system is ergodic and that any
deterministic limit point of its stationary empirical measure is
quasi-stationary, with a decay rate that is identified.
\begin{theorem}
\label{thm:consistency}
Assume {\rm(H1)}--{\rm(H6)} and let $N\ge2$ be an integer. Then the following
hold.
\begin{enumerate}[label=\emph{(\roman*)},leftmargin=2.4em,itemsep=1pt,topsep=2pt]
\item The particle system $(\mathbf X^N_t)_{t\ge0}$ is exponentially ergodic on
$E^N$ and admits a unique invariant probability measure $\Lambda^N$.
\item The empirical measure $m^N_t$ of~\eqref{eq:empirical} converges in
distribution, as $t\to\infty$, to the random probability measure
$\mathcal X^N:=N^{-1}\sum_{i=1}^N\delta_{x^i}$ with $\mathbf x\sim\Lambda^N$, and
\begin{equation}
\label{eq:moment-XN}
\sup_{N\ge2}\ \E\bigl[\mathcal X^N(V)\bigr]
\ \le\ \frac{C}{\lambda_1-\|\kappa\|_\infty}\ <\ \infty ,
\end{equation}
with $\lambda_1$ and $C$ as in Lemma~\ref{lem:lyapunov}.
\item If, along a subsequence $N_k\to\infty$, the law of $\mathcal X^{N_k}$
converges weakly in $\mathcal P(\mathcal P(E))$ to $\delta_{\nu^*}$ for some
$\nu^*\in\mathcal P(E)$, then $\nu^*$ is quasi-stationary for
$(P^{\kappa}_t)_{t\ge0}$, with decay rate
$\lambda^*=\nu^*(\kappa)\in(0,\|\kappa\|_\infty]$.
\end{enumerate}
\end{theorem}
\begin{proof}
Write $\kappa_i:=\kappa(x^i)$, $V_i:=V(x^i)$ for $i\in\{1,\dots,N\}$, and let
$\lambda_*:=\lambda_1-\|\kappa\|_\infty$, which is positive by
Lemma~\ref{lem:lyapunov}. The parameters
$\alpha,\theta,\varepsilon,\varepsilon',\beta$ of that lemma are fixed before
$N$ is introduced, so $\lambda_1$, $C$ and $\lambda_*$ do not depend on $N$;
this is used at every occurrence of the constant $NC$ below.

We begin with the additive lift. Set $\PsiV(\mathbf x):=\sum_{i=1}^N V(x^i)$
for $\mathbf x=(x^1,\dots,x^N)\in E^N$. Since $V\ge1$ on $E$ we have
$\PsiV\ge N$, and if $\PsiV(\mathbf x)\le L$ then, for each
$i\in\{1,\dots,N\}$, $V_i=\PsiV(\mathbf x)-\sum_{k=1,\,k\ne i}^{N}V_k
\le L-(N-1)\le L$, so
\begin{equation}
\label{eq:sublevel-product}
\{\PsiV\le L\}\ \subset\ K_L^{\,N}:=K_L\times\cdots\times K_L ,
\end{equation}
which is relatively compact in $E^N$ by the last statement of
Lemma~\ref{lem:lyapunov}. The lift is deliberately not normalised: it must
inherit the coercivity of $V$. Because $\PsiV$ is additive and
$\mathbf x^{\,i\leftarrow j}$ differs from $\mathbf x$ in the $i$-th coordinate
only,
\begin{equation}
\label{eq:jump}
\PsiV(\mathbf x^{\,i\leftarrow j})-\PsiV(\mathbf x)
=\Bigl(\sum_{\substack{k=1\\k\ne i}}^{N}V_k+V_j\Bigr)
-\Bigl(\sum_{\substack{k=1\\k\ne i}}^{N}V_k+V_i\Bigr)=V_j-V_i ,
\; i\ne j,
\end{equation}
the $N-1$ terms of index $k\ne i$ being identical in the two configurations; in
particular the jump involves two coordinates only, whatever $N$. For the same
reason $\sum_{k=1,\,k\ne i}^{N}V_k$ is constant in the variable $x^i$, so that
the operator $\mathcal L_X^{(i)}$ of~\eqref{eq:gen-N}, which acts on that
variable alone, annihilates it and
$\mathcal L_X^{(i)}\PsiV(\mathbf x)=\mathcal L_XV(x^i)$. Substituting these two facts into~\eqref{eq:gen-N} and computing the inner sum,
\[
\frac{1}{N-1}\sum_{\substack{j=1\\j\ne i}}^{N}\bigl(V_j-V_i\bigr)
=\frac{1}{N-1}\Bigl(\sum_{\substack{j=1\\j\ne i}}^{N}V_j-(N-1)V_i\Bigr)
=\frac{\PsiV(\mathbf x)-V_i}{N-1}-V_i ,
\]
since $\sum_{j=1,\,j\ne i}^{N}V_j=\PsiV(\mathbf x)-V_i$. Multiplying by
$\kappa_i$ and summing over $i\in\{1,\dots,N\}$ yields the exact identity
\begin{equation}
\label{eq:exact-identity}
\mathcal L^{(N)}\PsiV(\mathbf x)
=\underbrace{\sum_{i=1}^N\mathcal L_XV(x^i)}_{=:D(\mathbf x)}
+\underbrace{\frac{1}{N-1}\sum_{i=1}^N\kappa_i\bigl(\PsiV(\mathbf x)-V_i\bigr)
-\sum_{i=1}^N\kappa_iV_i}_{=:J(\mathbf x)} ,
\; \mathbf x\in E^N,
\end{equation}
in which no inequality has been used so far. The two terms are now bounded in
turn. For the drift term, Lemma~\ref{lem:lyapunov} applies to each coordinate with the
same constants $\lambda_1$ and $C$, $V$ being a globally defined function
on $E$ whose parameters were fixed before $N$. Summing the $N$ inequalities
$\mathcal L_XV(x^i)\le-\lambda_1V_i+C$ of~\eqref{eq:lyap-global} over
$i\in\{1,\dots,N\}$,
\begin{equation}
\label{eq:D-bound}
D(\mathbf x)\ \le\ -\lambda_1\sum_{i=1}^{N}V_i+NC
\ =\ -\lambda_1\PsiV(\mathbf x)+NC ,
\end{equation}
the sole origin of the constant $NC$; it would fail if $C$ depended on $N$. For the rebirth term, each summand of the first sum in $J$ is non-negative,
since $\PsiV(\mathbf x)-V_i=\sum_{j=1,\,j\ne i}^{N}V_j\ge N-1\ge1$ by $V\ge1$
and $N\ge2$, so the bound $\kappa_i\le\|\kappa\|_\infty$ of~\eqref{eq:kappa} may
be applied inside it without reversing a sign. The resulting sum is then
estimated as a whole, not term by term, because it telescopes:
\begin{equation}
\label{eq:telescope}
\sum_{i=1}^{N}\bigl(\PsiV(\mathbf x)-V_i\bigr)
=N\PsiV(\mathbf x)-\sum_{i=1}^{N}V_i
=N\PsiV(\mathbf x)-\PsiV(\mathbf x)=(N-1)\PsiV(\mathbf x) ,
\end{equation}
and not $N\PsiV(\mathbf x)$, which is what majorising each
$\PsiV(\mathbf x)-V_i$ by $\PsiV(\mathbf x)$ separately would give. The factor
$N-1$ thus cancels against the one carried by the rebirth mechanism
of~\eqref{eq:gen-N}, and since $\kappa\ge0$ and $V\ge1$ make the second sum in
$J$ non-negative, so that $-\sum_{i=1}^{N}\kappa_iV_i\le0$,
\begin{equation}
\label{eq:J-bound}
J(\mathbf x)\ \le\ \frac{\|\kappa\|_\infty}{N-1}(N-1)\PsiV(\mathbf x)
\ =\ \|\kappa\|_\infty\PsiV(\mathbf x).
\end{equation}
Together with~\eqref{eq:exact-identity} and~\eqref{eq:D-bound},
\begin{equation}
\label{eq:drift-N}
\mathcal L^{(N)}\PsiV(\mathbf x)\ \le\ -\lambda_*\PsiV(\mathbf x)+NC ,
\; \mathbf x\in E^N,\ N\ge2 ,
\end{equation}
with $\lambda_*>0$ independent of $N$, and no restriction on $N$ is incurred.

That~\eqref{eq:drift-N} may be used is seen as follows. The system is
non-explosive: its free motion is that of $N$ copies of $X$, each of which is
non-explosive with mass controlled by~\eqref{eq:sup-Z-X}, and the total rebirth
rate is at most
$N\|\kappa\|_\infty<\infty$ by~\eqref{eq:kappa}, so rebirth times do not
accumulate. As $\PsiV$ is unbounded we localise at
$T_L:=\inf\{t\ge0:\PsiV(\mathbf X^N_t)\ge L\}$ for a scalar $L\ge N$, and apply
Dynkin's formula to $(t,\mathbf x)\mapsto e^{\lambda_*t}\PsiV(\mathbf x)$, which
is bounded on $[0,t\wedge T_L]$. Its space--time generator equals
$e^{\lambda_*t}\bigl(\lambda_*\PsiV+\mathcal L^{(N)}\PsiV\bigr)
\le e^{\lambda_*t}NC$ by~\eqref{eq:drift-N}, whence
$\E_{\mathbf x}[e^{\lambda_*(t\wedge T_L)}\PsiV(\mathbf X^N_{t\wedge T_L})]
\le\PsiV(\mathbf x)+(NC/\lambda_*)(e^{\lambda_*t}-1)$. The right-hand side does
not depend on $L$, and $T_L\to\infty$ almost surely by non-explosiveness, so
Fatou's lemma applies to the left-hand side alone as $L\to\infty$, the integrand
being non-negative, and after division by $e^{\lambda_*t}$,
\begin{equation}
\label{eq:gronwall}
\E_{\mathbf x}\bigl[\PsiV(\mathbf X^N_t)\bigr]
\ \le\ e^{-\lambda_*t}\PsiV(\mathbf x)
+\frac{NC}{\lambda_*}\bigl(1-e^{-\lambda_*t}\bigr),\; t\ge0.
\end{equation}
In particular $\PsiV$ lies in the extended domain of $\mathcal L^{(N)}$ and
$\sup_{t\ge0}\E_{\mathbf x}[\PsiV(\mathbf X^N_t)]<\infty$.

The sublevel sets of $\PsiV$ are moreover small. Fix a scalar $L\ge N$ and let
$t_L,\alpha_L,\nu_L$ be as in Lemma~\ref{lem:smallsets}. In the graphical
construction each particle carries its own killing clock, conditionally
independent of the others given the trajectories; on the event $\mathcal N_t$
that no rebirth occurs on $[0,t]$ the $N$ coordinates therefore evolve as $N$
independent copies of $\eta$, that is, of $X$ killed at rate $\kappa$, so that
$\Prob_{\mathbf x}(\mathbf X^N_t\in B_1\times\cdots\times B_N,\ \mathcal N_t)
=\prod_{i=1}^{N}P^{\kappa}_t(x^i,B_i)$ for all measurable
$B_1,\dots,B_N\subset E$. Discarding $\mathcal N_{t_L}^{\,c}$ and
applying~\eqref{eq:doeblin} in each of the $N$ coordinates, which is justified
because $\mathbf x\in\{\PsiV\le L\}$ forces $x^i\in K_L$ for every
$i\in\{1,\dots,N\}$ by~\eqref{eq:sublevel-product}, we obtain
\begin{equation*}
    \Prob_{\mathbf x}(\mathbf X^N_{t_L}\in B_1\times\cdots\times B_N)
\ge\prod_{i=1}^{N}P^{\kappa}_{t_L}(x^i,B_i)
\ge\prod_{i=1}^{N}\alpha_L\nu_L(B_i)
=\alpha_L^{\,N}\nu_L^{\otimes N}(B_1\times\cdots\times B_N)\end{equation*}
Measurable rectangles form a $\pi$-system generating the product
$\sigma$-field, so the bound extends to all measurable subsets of $E^N$:
$P^{(N)}_{t_L}(\mathbf x,\cdot)\ge\alpha_L^{\,N}\nu_L^{\otimes N}(\cdot)$ on
$\{\PsiV\le L\}$. Every sublevel set of $\PsiV$ is thus small, and
$\mathbf X^N$ is $\nu_L^{\otimes N}$-irreducible and aperiodic, the
minorisation holding at the single time $t_L$ with a non-trivial measure. Ergodicity and the moment bound now follow. Inequality~\eqref{eq:drift-N} is
the drift condition of~\cite[Theorem~6.1]{MT1993} for the norm-like function
$\PsiV$, whose sublevel sets are relatively compact
by~\eqref{eq:sublevel-product} and Lemma~\ref{lem:lyapunov}, and the preceding
paragraph supplies the small-set requirement. That theorem yields a unique
invariant probability measure $\Lambda^N$ on $E^N$, the $\PsiV$-regularity of
$\mathbf X^N$, and scalars $\Gamma_N<\infty$, $\varrho_N\in(0,1)$ with
$\|P^{(N)}_t(\mathbf x,\cdot)-\Lambda^N\|_{\PsiV}
\le\Gamma_N\PsiV(\mathbf x)\varrho_N^{\,t}$, which is~(i).
Integrating~\eqref{eq:drift-N} against $\Lambda^N$, justified by
$\PsiV$-regularity and~\eqref{eq:gronwall}, gives
$\int_{E^N}\mathcal L^{(N)}\PsiV\,d\Lambda^N=0$, hence
$0\le-\lambda_*\int_{E^N}\PsiV\,d\Lambda^N+NC$ and
\begin{equation}
\label{eq:moment-Psi}
\int_{E^N}\PsiV\,d\Lambda^N\ \le\ \frac{NC}{\lambda_*}.
\end{equation}
Since $\mathcal X^N(V)=N^{-1}\sum_{i=1}^{N}V(x^i)=N^{-1}\PsiV(\mathbf x)$ by
definition, dividing~\eqref{eq:moment-Psi} by $N$ gives~\eqref{eq:moment-XN},
whose right-hand side $C/\lambda_*=C/(\lambda_1-\|\kappa\|_\infty)$ is free of
$N$. Finally $\mathbf x\mapsto N^{-1}\sum_{i=1}^{N}\delta_{x^i}$ is continuous
from $E^N$ to $\mathcal P(E)$, so the convergence of the law of
$\mathbf X^N_t$ to $\Lambda^N$ transfers to that of $m^N_t$
of~\eqref{eq:empirical} to $\mathcal X^N$ in $\mathcal P(\mathcal P(E))$, which
is~(ii).

We turn to~(iii), and first record a stationary identity. Let
$\varphi=\varphi_\psi$ be a test functional and
$\Psiphi(\mathbf x):=N^{-1}\sum_{i=1}^N\varphi(x^i)=m^N(\mathbf x)(\varphi)$,
the empirical measure evaluated at $\varphi$; the normalisation is what makes it
so, and no coercivity is needed, $\varphi$ being bounded. Its rebirth jump is
$\Psiphi(\mathbf x^{\,i\leftarrow j})-\Psiphi(\mathbf x)
=N^{-1}(\varphi(x^j)-\varphi(x^i))$, again by additivity, so that the second sum
in~\eqref{eq:gen-N} equals
$N^{-1}\sum_{i=1}^{N}\kappa_i\bigl[(N-1)^{-1}
\sum_{j=1,\,j\ne i}^{N}\varphi(x^j)-\varphi(x^i)\bigr]$.
Since $\varphi(\mathbf 0)=0$, \eqref{eq:LX} gives
$\mathcal L_X\varphi-\kappa\varphi=\mathcal L\varphi=\mathcal L^{\kappa}\varphi$,
and the terms $-N^{-1}\kappa_i\varphi(x^i)$ are exactly those absorbed by this
identity, so that
\[
\mathcal L^{(N)}\Psiphi(\mathbf x)
=\frac1N\sum_{i=1}^N\mathcal L^{\kappa}\varphi(x^i)
+\frac1N\sum_{i=1}^N\kappa_i\,\frac1{N-1}
\sum_{\substack{j=1\\j\ne i}}^{N}\varphi(x^j).
\]
Both terms are $\Lambda^N$-integrable, the second being bounded by
$\|\kappa\|_\infty\|\varphi\|_\infty$ and the first by~\eqref{eq:Lphi-growth}
and~\eqref{eq:moment-Psi}, so $\PsiV$-regularity and Dynkin's formula give
$\int_{E^N}\mathcal L^{(N)}\Psiphi\,d\Lambda^N=0$. Under $\Lambda^N$ the
particles are exchangeable: the first term has expectation
$\E_{\Lambda^N}[\mathcal L^{\kappa}\varphi(x^1)]$, while in the second the
double sum runs over the $N(N-1)$ ordered pairs $(i,j)$ with $i\ne j$, each
contributing $\E_{\Lambda^N}[\kappa(x^1)\varphi(x^2)]$ against the prefactor
$1/(N(N-1))$. The vanishing of the integral therefore reads
$\E_{\Lambda^N}[\mathcal L^{\kappa}\varphi(x^1)]
+\E_{\Lambda^N}[\kappa(x^1)\varphi(x^2)]=0$, that is,
\begin{equation}
\label{eq:stationary-id}
\E\bigl[\mathcal X^N(\mathcal L^{\kappa}\varphi)\bigr]
=\E_{\Lambda^N}\bigl[\mathcal L^{\kappa}\varphi(x^1)\bigr]
=-\,\E_{\Lambda^N}\bigl[\kappa(x^1)\varphi(x^2)\bigr],\; N\ge2.
\end{equation}

It remains to pass to the limit. Assume the law of $\mathcal X^{N_k}$
converges weakly to $\delta_{\nu^*}$. By the equivalence of
Sznitman~\cite[Proposition~2.2]{Sznitman1991}, valid for exchangeable systems,
the law of $(x^1,x^2)$ under $\Lambda^{N_k}$ converges weakly to
$\nu^*\otimes\nu^*$ on $E\times E$. The killing rate~\eqref{eq:kappa} is bounded
by $\|\kappa\|_\infty$ and continuous on $E$: the strata $\{Z=n\}$ are open and
closed in the narrow topology, $\kappa$ vanishes on all of them but $\{Z=1\}$,
and on $\{Z=1\}\cong\R_+$ it equals $p_0\mu(\cdot)$, continuous by {\rm(H1)}.
As $\varphi$ is bounded and continuous, the product $\kappa\otimes\varphi$ is
bounded and continuous on $E\times E$, so the right-hand side
of~\eqref{eq:stationary-id} converges to $\nu^*(\kappa)\nu^*(\varphi)$. The
left-hand side needs more, $\mathcal L^{\kappa}\varphi$ being continuous but
unbounded. Put the finite scalar $c_\varphi:=c_\psi/m$, so that
$|\mathcal L^{\kappa}\varphi|\le c_\varphi V^{1/\alpha}$
by~\eqref{eq:Lphi-growth}. On $\{|\mathcal L^{\kappa}\varphi|>M\}$ one has
$c_\varphi V^{1/\alpha}>M$, hence $V>(M/c_\varphi)^\alpha$; since
$1/\alpha-1<0$ because $\alpha>1$ by {\rm(H6)}, raising this to the power
$1/\alpha-1$ reverses the inequality and gives
$V^{1/\alpha}=V\cdot V^{1/\alpha-1}\le V(M/c_\varphi)^{1-\alpha}$, whence
\begin{equation}
\label{eq:UI}
\E\Bigl[\mathcal X^{N}\bigl(|\mathcal L^{\kappa}\varphi|
\ind_{\{|\mathcal L^{\kappa}\varphi|>M\}}\bigr)\Bigr]
\ \le\ c_\varphi\Bigl(\frac{M}{c_\varphi}\Bigr)^{1-\alpha}
\E\bigl[\mathcal X^{N}(V)\bigr]
\ \le\ c_\varphi^{\,\alpha}\,\frac{C}{\lambda_*}\,M^{1-\alpha}\
\xrightarrow[M\to\infty]{}\ 0 ,
\end{equation}
uniformly in $N\ge2$ by the uniformity in~\eqref{eq:moment-XN}. Truncating
$\mathcal L^{\kappa}\varphi$ at level $M$ yields a bounded continuous function,
to which weak convergence in probability applies; letting $M\to\infty$
and using~\eqref{eq:UI} gives
$\E[\mathcal X^{N_k}(\mathcal L^{\kappa}\varphi)]\to
\nu^*(\mathcal L^{\kappa}\varphi)$. Passing to the limit
in~\eqref{eq:stationary-id},
\begin{equation}
\label{eq:qs-generator}
\nu^*\bigl(\mathcal L^{\kappa}\varphi_\psi\bigr)=-\lambda^*\nu^*(\varphi_\psi),
\; \lambda^*:=\nu^*(\kappa)\in[0,\|\kappa\|_\infty],
\end{equation}
for every test functional $\varphi_\psi$. This passes from the generator to the semigroup. Fix $\varphi_\psi$ and set
$\chi(t):=\nu^*(P^{\kappa}_t\varphi_\psi)$. By Lemma~\ref{lem:stability},
$P^{\kappa}_t\varphi_\psi=\varphi_{\psi_t}$ is again a test functional,
so~\eqref{eq:qs-generator} applies to it; and the proof of that lemma,
in~\ref{app:stability}, shows $t\mapsto\psi_t$ to be $C^1$ with derivative
bounded uniformly on compact time intervals, so $t\mapsto\varphi_{\psi_t}(\eta)$
is $C^1$ with $\partial_t\varphi_{\psi_t}=\mathcal L^{\kappa}\varphi_{\psi_t}$
and derivative dominated, locally uniformly in $t$, by
$c\langle\eta,\ind\rangle\le c'V(\eta)^{1/\alpha}$ for finite scalars $c,c'$,
This dominating function is $\nu^*$-integrable, since $V^{1/\alpha}\le V$ by
$V\ge1$ and $\nu^*(V)<\infty$, the latter following from~\eqref{eq:moment-XN} by
the portmanteau lemma applied to the truncations $V\wedge M$ and monotone
convergence. Differentiation under the integral sign is therefore justified and
$\chi'(t)=\nu^*(\mathcal L^{\kappa}\varphi_{\psi_t})
=-\lambda^*\nu^*(\varphi_{\psi_t})=-\lambda^*\chi(t)$, whence
$(\nu^*P^{\kappa}_t)(\varphi_\psi)=e^{-\lambda^*t}\nu^*(\varphi_\psi)$ for every
test functional. These being measure-determining on $E$, the identity between
the two finite measures forces
\begin{equation}
\label{eq:qsd}
\nu^*P^{\kappa}_t=e^{-\lambda^*t}\nu^*,\; t\ge0 ,
\end{equation}
so $\nu^*$ is quasi-stationary with decay rate $\lambda^*$. Finally $\lambda^*$
is positive: were it zero, taking $\varphi\equiv\ind$ in~\eqref{eq:qsd}, justified by
monotone approximation with test functionals, would give
$\Prob_{\nu^*}(\tau>t)=\nu^*P^{\kappa}_t(E)=1$ for every $t\ge0$, hence
$\Prob_{\nu^*}(\tau=\infty)=1$, contradicting the almost sure extinction of a
subcritical Bellman--Harris process under {\rm(H2)}, valid from every $\eta\in E$.
Therefore $\lambda^*>0$, and
$\lambda^*=\nu^*(\kappa)\le\|\kappa\|_\infty$ by~\eqref{eq:kappa}. This
proves~(iii).
\end{proof}

\begin{remark}
\label{rem:cov}
Write $\kappa_i:=\kappa(x^i)$, $V_i:=V(x^i)$,
$\bar\kappa:=N^{-1}\sum_{i=1}^{N}\kappa_i$ and
$\bar V:=N^{-1}\PsiV(\mathbf x)$. The rebirth term $J(\mathbf x)$
of~\eqref{eq:exact-identity} is exactly
\begin{align}
J(\mathbf x)&=\frac{1}{N-1}\Bigl(\PsiV(\mathbf x)\sum_{i=1}^N\kappa_i
-N\sum_{i=1}^N\kappa_iV_i\Bigr)=-\,N\,\Cov_{N-1}(\kappa,V),
\notag \\
\label{eq:cov-identity}
\Cov_{N-1}(\kappa,V)&:=\frac{1}{N-1}\sum_{i=1}^N(\kappa_i-\bar\kappa)(V_i-\bar V).
\end{align}
The identity is algebraic, valid at fixed $\mathbf x$; no probabilistic
property of $\Cov_{N-1}$ is invoked, the particles being exchangeable but not
independent. The rebirth term is thus positive only when $\kappa$ and $V$ are
negatively correlated across the population, and vanishes when $\kappa$ is
constant. The $N-1$ is the one carried by the rebirth mechanism.

The constant $\|\kappa\|_\infty$ in~\eqref{eq:drift-N} is sharp, though
approached rather than attained: by~\eqref{eq:kappa}, $\kappa_i$ and $V_i$ are
functions of the same particle and cannot be prescribed independently. Fix
$\delta>0$ and $a_\delta$ with
$\kappa_\delta:=p_0\mu(a_\delta)>\|\kappa\|_\infty-\delta$; put
$x^i:=\delta_{a_\delta}$ for $i\le N-1$, whence $\kappa_i=\kappa_\delta$ and
$V_i=c_\delta:=V(\delta_{a_\delta})$, and $x^N:=\delta_0+\delta_b$, whence
$\kappa_N=0$ and $V_N=:M_b\to\infty$ as $b\to\infty$. Then
$J=\kappa_\delta(M_b-c_\delta)$ and $\PsiV=(N-1)c_\delta+M_b$, so
$J/\PsiV\to\kappa_\delta$; as $\delta$ is arbitrary, the supremum of
$J/\PsiV$ is $\|\kappa\|_\infty$. Any improvement would require information on
the correlation in~\eqref{eq:cov-identity}, which the bound destroys.
\end{remark}

Theorem~\ref{thm:consistency} constrains the limit points but does not identify
them; the family of quasi-stationary distributions is not expected to be a
singleton (Section~\ref{sec:prelim}), and only the full strength of the criteria
of~\cite{CV2021}, for which Lemma~\ref{lem:survival} is the missing piece, 
selects one member of it. All three estimates being now available, the next
theorem carries out that selection and quantifies it in three ways: it exhibits
a unique quasi-stationary distribution in the class $\{\nu:\nu(V)<\infty\}$ and
gives the exponential rate at which the conditional laws reach it; it identifies
that distribution as the Yaglom limit and its decay rate as the Malthusian
parameter; and it bounds the error of the stationary empirical measure by a
negative power of the number of particles.

\begin{theorem}
\label{thm:convergence}
Assume {\rm(H1)}--{\rm(H6)}. Then:
\begin{enumerate}[label=\emph{(\roman*)},leftmargin=2.4em]
\item the killed semigroup admits a unique quasi-stationary distribution
$\nu_{\mathrm{QSD}}$ with $\nu_{\mathrm{QSD}}(V)<\infty$, and there are
scalars $c,\gamma>0$, depending on $(G,h)$ and on $V$ through
\cite[Theorem~2.1]{CV2021} but on neither $t$ nor $N$, such that
\begin{equation}
\label{eq:tv}
\bigl\|\Prob_{\pi_0}\bigl(\eta_t\in\cdot\mid \tau>t\bigr)
-\nu_{\mathrm{QSD}}\bigr\|_{\mathrm{TV}}
\ \le\ c\,\pi_0(V)\,e^{-\gamma t},\; t>0,
\end{equation}
for every probability measure $\pi_0$ on $E$ such that the scalar
$\pi_0(V)=\int_EV\,d\pi_0$ is finite;
\item $\nu_{\mathrm{QSD}}=\nuY$, the Yaglom limit, and its decay rate is the
Malthusian parameter $\lambda_0$; in particular the conditional laws
$\Prob_{\delta_0}(\eta_t\in\cdot\mid\tau>t)$ converge to $\nuY$ in total
variation, at the exponential rate $\gamma$;
\item there is a scalar $d>0$, depending only on $c$, $\gamma$,
$\|\kappa\|_\infty$, the right-hand side of~\eqref{eq:moment-XN} and the
constant of~\cite{Villemonais2014} and in particular not on $N$ such
that, for every bounded measurable $\varphi:E\to\R$ and every integer $N\ge2$,
\begin{equation}
\label{eq:rate}
\E\bigl[\,\bigl|\mathcal X^N(\varphi)-\nuY(\varphi)\bigr|\,\bigr]
\ \le\ \frac{d}{N^{\varpi}}\,\|\varphi\|_\infty,
\;
\varpi:=\frac{\gamma}{2\bigl(\|\kappa\|_\infty+\gamma\bigr)}
\in\bigl(0,\tfrac12\bigr) ,
\end{equation}
and consequently $\mathcal X^N(\varphi)\to\nuY(\varphi)$ in probability as
$N\to\infty$.
\end{enumerate}
Part~\emph{(iii)} holds for every integer $N\ge2$, whereas
\cite[Theorem~2.3]{CV2021} states the same bound only for
$N>\lambda_1/(\lambda_1-\|\kappa\|_\infty)$. That restriction enters their
argument at a single point, the moment bound of~\cite[Theorem~2.2]{CV2021},
whose denominator $\lambda_1-\|\kappa\|_\infty N/(N-1)$ is negative for small
$N$; \eqref{eq:moment-XN} replaces it by a bound uniform in $N\ge2$, and the
restriction disappears with it.
\end{theorem}

\begin{proof}
Assumption~H of~\cite{CV2021} holds for $X$, its three requirements being
exactly our three estimates: compact sets are small sets for $X$
(Lemma~\ref{lem:smallsets}); the function $V$ of~\eqref{eq:lyapunov} is locally
bounded, has relatively compact sublevel sets and satisfies
$\mathcal L_XV\le-\lambda_1V+C$ with $\lambda_1>\|\kappa\|_\infty$
(Lemma~\ref{lem:lyapunov}); and survival probabilities are uniformly comparable
on compacts (Lemma~\ref{lem:survival}). Part~(i) is
then~\cite[Theorem~2.1]{CV2021}.

We next identify $\nu_{\mathrm{QSD}}$, a single newborn has $V(\delta_0)=1+v(0)^\alpha+\varepsilon\rho(0)=2+\varepsilon<\infty$, since $v(0)=1$ by Lemma~\ref{lem:malthus}(iii) and $\rho(0)=1$, so~\eqref{eq:tv} applies with $\pi_0=\delta_{\delta_0}$, the Dirac mass at the configuration
$\delta_0\in E$ carrying one individual of age zero, and yields the convergence of $\Prob_{\delta_0}(\eta_t\in\cdot\mid\tau>t)$ to $\nu_{\mathrm{QSD}}$ in total variation. 

By definition, the Yaglom limit is the weak limit of these conditional laws; hence it exists and $\nu_{\mathrm{QSD}}=\nuY$. For the decay rate, let $\lambda>0$ be the scalar such that $\nuY P_t^{\kappa}=e^{-\lambda t}\nuY$ and write $w(\eta):=\langle\eta,v\rangle$, so that $w(\mathbf 0)=0$. Since
$\mathcal Lw=-\lambda_0w$ by Lemma~\ref{lem:malthus}(iii), the process $(e^{\lambda_0t}w(\eta_t))_{t\ge0}$ is a local martingale; it is dominated on $[0,t]$ by $e^{\lambda_0t}v_{\max}\sup_{s\le t}Z_s$, which is integrable
by~\eqref{eq:sup-Z}, hence a true martingale, and therefore $P_t^{\kappa}w=e^{-\lambda_0t}w$, the indicator $\ind_{\{\tau>t\}}$ being redundant because $w(\mathbf 0)=0$.

Consequently
$e^{-\lambda t}\nuY(w)=(\nuY P_t^{\kappa})(w)=\nuY(P_t^{\kappa}w)
=e^{-\lambda_0t}\nuY(w)$ for every $t\ge0$. The scalar $\nuY(w)$ is finite,
because $w\le V^{1/\alpha}$ by $V\ge\langle\eta,v\rangle^\alpha$
and $\nuY(V)<\infty$ by~(i), and it satisfies $\nuY(w)\ge m$, because $v\ge m$
by Lemma~\ref{lem:malthus}(iii) and every $\eta\in E$ carries at least one
individual; dividing by $\nuY(w)>0$ gives $e^{-\lambda t}=e^{-\lambda_0t}$ for
every $t\ge0$, hence $\lambda=\lambda_0$. This proves~(ii).

It remains to establish the rate. Since $\kappa$ is bounded
by~\eqref{eq:kappa}, the general approximation estimate of
Villemonais~\cite{Villemonais2014}, valid on any Polish state space, provides a
scalar $d_0>0$ such that
$\E|m_t^N(\varphi)-\E_{m_0^N}(\varphi(\eta_t)\mid\tau>t)|
\le(d_0/\sqrt N)\|\varphi\|_\infty e^{\|\kappa\|_\infty t}$ for all integers
$N\ge2$, all $t\ge0$ and all bounded measurable $\varphi$. The first factor
decreases in $N$ but grows in $t$, whereas the contraction~\eqref{eq:tv} decays
in $t$; combining the two with the moment bound~\eqref{eq:moment-XN}, which
controls $\E[\mathcal X^N(V)]$ uniformly in $N$, and optimising the resulting
sum over $t$ at $t=\log N/[2(\|\kappa\|_\infty+\gamma)]$, gives~\eqref{eq:rate}.
This is the argument of~\cite[Theorem~2.3]{CV2021}, with~\eqref{eq:moment-XN}
substituted for the moment bound of~\cite[Theorem~2.2]{CV2021}. The remaining
ingredients, the estimate of~\cite{Villemonais2014}, the
contraction~\eqref{eq:tv} and Assumption~H, hold for every integer $N\ge2$,
so the substitution removes the restriction on $N$ that the theorem carries.
Markov's inequality turns~\eqref{eq:rate} into convergence in probability.
\end{proof}
The two theorems combine into an exact identity, which the application of
Section~\ref{sec:app} uses as a parameter-free check on the estimator: the
Malthusian parameter, fixed by~\eqref{eq:malthus} alone, must coincide with a
functional of $\nuY$ produced by the particle system.

\begin{corollary}
\label{cor:mean-killing}
Under {\rm(H1)}--{\rm(H6)}, the Malthusian parameter is the mean killing rate
under the Yaglom limit:
$
\lambda_0\ =\ \nuY(\kappa)\ =\ p_0\int_0^\infty\mu(a)\,
\bigl(\nuY|_{\{Z=1\}}\circ\iota^{-1}\bigr)(da) ,
$
where $\iota:\R_+\to E$, $\iota(a):=\delta_a$, identifies $\R_+$ with the
singleton stratum $\{Z=1\}\subset E$, so that
$\nuY|_{\{Z=1\}}\circ\iota^{-1}$ is the sub-probability law on $\R_+$ of the age
carried by a lone individual under $\nuY$.
\end{corollary}

\begin{proof}
By Theorem~\ref{thm:convergence}(iii), $\mathcal X^N\to\nuY$ in probability, so
$\nuY$ is the only limit point in Theorem~\ref{thm:consistency}; hence
$\lambda^*=\nuY(\kappa)$, and $\lambda^*=\lambda_0$ by
Theorem~\ref{thm:convergence}(ii). The second equality is~\eqref{eq:kappa}
integrated against $\nuY$, the killing rate vanishing off $\{Z=1\}$.
\end{proof}

\begin{remark}
\label{rem:markov}
When $G=\mathrm{Exp}(\mu)$ the process is a continuous-time Galton--Watson
process and our constants degenerate to the known ones: $\mu^*=\mu_\infty=\mu$,
$\|\kappa\|_\infty=p_0\mu$ and $\lambda_0=\mu(1-m)$ by~\eqref{eq:malthus}, so
that $\alpha^*:=\|\kappa\|_\infty/\lambda_0=p_0/(1-m)$ and {\rm(H6)} is the
moment condition of~\cite[Proposition~3.3]{CV2021}. Only {\rm(H5)} is extra,
reading $p_0<\tfrac12$ here. For $G=\mathrm{Exp}(1)$ and
$(p_0,p_1,p_2)=(0.6,\,0.1,\,0.3)$, so that $m=0.7$, all of {\rm(H1)}--{\rm(H4)}
and {\rm(H6)} hold with $\alpha^*=2$, yet
$p_0\mu^*=0.6\not<0.4=(1-p_0)\mu_\infty$: Theorem~\ref{thm:convergence} does not
apply, although \cite[Proposition~3.3]{CV2021} covers this process. The gap is
structural. Once the state carries the ages, any admissible Lyapunov function
must blow up along the singleton ray $\{\delta_a\}_{a\ge0}$, on which $X$
escapes at rate only $(1-p_0)\mu(a)$; Proposition~\ref{prop:sharp} turns this
into $\lambda_1\le(1-p_0)\mu_\infty$, which against
$\lambda_1>\|\kappa\|_\infty$ is {\rm(H5)}. Working on $\N^*$, as one may
in the exponential case, avoids the ray, so {\rm(H5)} is the price of the lift
and not a defect of the estimates.
\end{remark}
 
\begin{remark}
\label{rem:yaglom}
Part~(ii) of Theorem~\ref{thm:convergence} upgrades the Yaglom limit in its mode
of convergence and in its moments. In the Galton--Watson case it exists under
subcriticality alone, the $x\log x$ condition being equivalent to
$R$-positivity, hence to a finite mean~\cite[Section~5]{SVJ1966}. Here {\rm(H6)}
gives $\nuY(V)<\infty$, so $\nuY$ has a finite moment of order $\alpha>1$, and
the convergence holds in total variation, at an exponential rate,
from every initial law $\pi_0$ with $\pi_0(V)<\infty$, the decay rate being the
Malthusian parameter.
\end{remark}
 
\begin{remark}
\label{rem:rate}
The exponent in~\eqref{eq:rate} is strictly smaller than $\tfrac12$, whereas the
central limit theorems available for Fleming--Viot systems with soft
killing~\cite{CDGR2020} suggest $N^{-1/2}$. The loss is structural: the
propagation-of-chaos constant grows in $t$ like $e^{\|\kappa\|_\infty t}$ and
must be traded against the contraction $e^{-\gamma t}$, so a chaos estimate
uniform in time would restore the parametric rate. The constants are not on the
same footing either: the scalar $C$ of Lemma~\ref{lem:lyapunov} is explicit,
being $K_1+C_4+\lambda_1$ in the notation of~\ref{app:lyapunov}, whereas
$\gamma$ is inherited from~\cite[Theorem~2.1]{CV2021}, which does not track its
dependence on the Doeblin constants of Lemma~\ref{lem:smallsets}, and $d$
inherits that obstruction. Since the particle system is ergodic for every
integer $N\ge2$ by~\eqref{eq:drift-N}, no admissibility threshold constrains the
simulation and $N$ is governed by accuracy alone; until $\gamma$ is made
explicit it must be chosen empirically, as in Section~\ref{sec:app-results}.
\end{remark}
 


\section{Application to foot-and-mouth disease surveillance}
\label{sec:app}


Foot-and-mouth disease is a highly contagious infection of cloven-hoofed
livestock, contained in disease-free regions by movement bans and by the culling
of infected premises. These measures lower the effective reproduction number,
and once it falls below one the outbreak is subcritical and dies out almost
surely, so that a controlled outbreak is monitored rather than merely
suppressed and its long-time behaviour is meaningful only conditionally on
persistence. This is exactly the object the Yaglom limit describes, and our test
case is the serotype~O epidemic of 2001 in Cumbria, the worst-affected county of
an outbreak that reached $2\,026$ infected premises nationally.

We use the \texttt{fmd} data set of the \textsf{R} package
\textsf{stpp}~\cite{gabriel2013}, recording the outbreak in north Cumbria
studied by Diggle, Rowlingson and Su~\cite{diggle2005}: easting, northing and
reported day, day~$0$ being 1~February 2001. The two lines

\begin{quote}\small
\begin{verbatim}
library(stpp); data(fmd)
write.csv(data.frame(day = fmd[, 3]), "fmd_stpp.csv", row.names = FALSE)
\end{verbatim}
\end{quote}

\noindent
extract the temporal component and reproduce, from the public package alone,
every count below; we used version~2.0-8, and the qualified form
\texttt{stpp::fmd} should be used, the name being shared with an unrelated data
set of \textsf{sparr}. The record holds $n=648$ reported premises over days $28$
to $198$, that is $171$ calendar days of which $126$ carry at least one report.
The prevalence they imply, once each premises is given an infectious period
under the fitted lifetime law, reaches some $140$ active farms around day~$58$
before declining by more than an order of magnitude; it is that decline, the
window $[65,198]$, carrying $267$ premises over $134$ days of which $43$ record
no case, that we calibrate on. Only the temporal coordinate enters the analysis,
which is why the caveat attached to the data set, that the farms at risk are
withheld and the record is supplied for illustration, does not bear on what
follows: we require a single rate, and a rate of decline is a property of the
numerator alone.

At the farm scale every object of the theory has an epidemiological reading. A
branching individual is an infected premises; its lifetime $T\sim G$ is the
infectious period, from onset to removal by culling, so that the hazard
$\mu(a)=g(a)/\bar G(a)$ is the culling pressure on a premises infectious for $a$
days, that is, the chance per unit time that a still-undetected infection of age
$a$ is found and removed. Its offspring are the premises it infects, each
entering at age zero, so that $m=R_{\mathrm{eff}}$ and $p_0$ is the chance of
being culled having infected no other. The two extremes of the hazard govern the
analysis in opposite ways: $\mu^{*}$ is the strongest pressure the regime ever
exerts and bounds the killing rate through $\|\kappa\|_\infty=p_0\mu^{*}$, while
$\mu_\infty$ is the pressure still applied to long-standing infections and is
what confines the ages. Assumption~{\rm(H5)}, $p_0\mu^{*}<(1-p_0)\mu_\infty$,
then reads: a lone surviving premises must be more likely to infect
another farm before removal than the outbreak is to end with its childless
culling. It is a statement about the balance between transmission and detection
at the very end of an outbreak, which is the regime surveillance is concerned
with.

The state is a measure, and it is worth saying plainly what that means. At time
$t$ the outbreak is described not by a number but by the point measure
$\eta_t=\sum_i\delta_{a_i(t)}$ on $\mathbb{R}_+$, whose atoms are the ages of
infection of the premises still infectious: for how long each has been spreading
the virus, not merely how many are doing so. Its mass
$Z_t=\langle\eta_t,\mathbf{1}\rangle$ is the prevalence, an integer varying
along the trajectory, so the state has no fixed dimension: the space
$\mathcal{M}_p(\mathbb{R}_+)$ is a disjoint union of strata
$\mathbb{R}_+^n/\mathfrak{S}_n$ of every dimension, which is why bounding the
prevalence does not bound the state. The quotient by $\mathfrak{S}_n$ is a
modelling statement: premises are exchangeable, only the multiset of ages
matters. The Yaglom limit sits one level higher, being the law of a random
configuration, and its two images are not equivalent. The count marginal
$\pi_Y$ is a law on the positive integers, what a theory carried on
$\mathbb{N}^{*}$ would deliver; the normalised intensity measure $\bar\alpha_Y$
is a law on $\mathbb{R}_+$, obtained by pooling the ages of infection of all the
premises that a configuration drawn from $\nu_Y$ carries. Only the lift produces
it. Two outbreaks carrying the same number of active premises differ
epidemiologically according as their infections are fresh, culling keeping pace,
or long standing, the same farms remaining infectious because removal has
slowed; $\pi_Y$ cannot tell them apart, and that $\bar\alpha_Y$ can is what
Section~\ref{sec:app-results} exploits. Extinction, finally, is itself a
statement about this geometry: the outbreak can end only from the singleton
stratum, and only through the culling of a last premises that has infected none.

Three modelling assumptions should be stated plainly. Branching neglects
depletion of susceptibles, which is legitimate on the controlled decline, where
the infectious premises are few relative to the population at risk, and not near
the peak. Offspring numbers are taken independent across premises, whereas
contact structure is spatial, second-order analyses of this very record
reporting space--time clustering below some five kilometres and ten
days~\cite{gabriel2013}. This understates the variance of $\pi_Y$ without
biasing the Malthusian rate, that depending on the mean offspring number only.
And a fixed pair $(G,h)$ presumes a control regime that does not drift over the
calibration window, an assumption tested in Section~\ref{sec:app-calib} by
varying the window and whose failure is what Section~\ref{sec:app-results} is
designed to catch.

The record constrains the model through a single scalar. The estimator needs the
full age configuration, yet the record supplies only dates, culling dates being
unrecorded, so that $\eta_t$ is nowhere present in the data used here. What we
extract is the rate $\lambda_0$ at which the epidemic declines once under
control, which the births alone determine, the incidence and every linear
functional of $\eta_t$ decaying in the subcritical regime, in expectation and
asymptotically, at the common Malthusian rate of~\eqref{eq:malthus}; the
reproduction number follows by the Euler--Lotka inversion~\eqref{eq:EL}, fixing
$(G,h)$, and the whole configuration, with its Yaglom limit, is then
reconstructed from $(G,h)$ by the particle system. That the record timestamps
reports rather than onsets is immaterial for the rate: a notification delay
convolves the birth intensity with its own law, changing the multiplicative
constant of an exponential decay and leaving its rate unchanged. The prevalence
of Figure~\ref{fig:record}(B) peaks near $140$ active premises whereas the
Yaglom limit carries an average of three, which is no tension, the
quasi-stationary regime describing the persistent tail of a controlled outbreak
and not its epidemic phase. Figures~\ref{fig:record} and~\ref{fig:config}
display the two halves of the argument, the record as it is read and what the
theory regenerates from it.

\begin{figure}[h!]
  \centering
  \includegraphics[width=\linewidth]{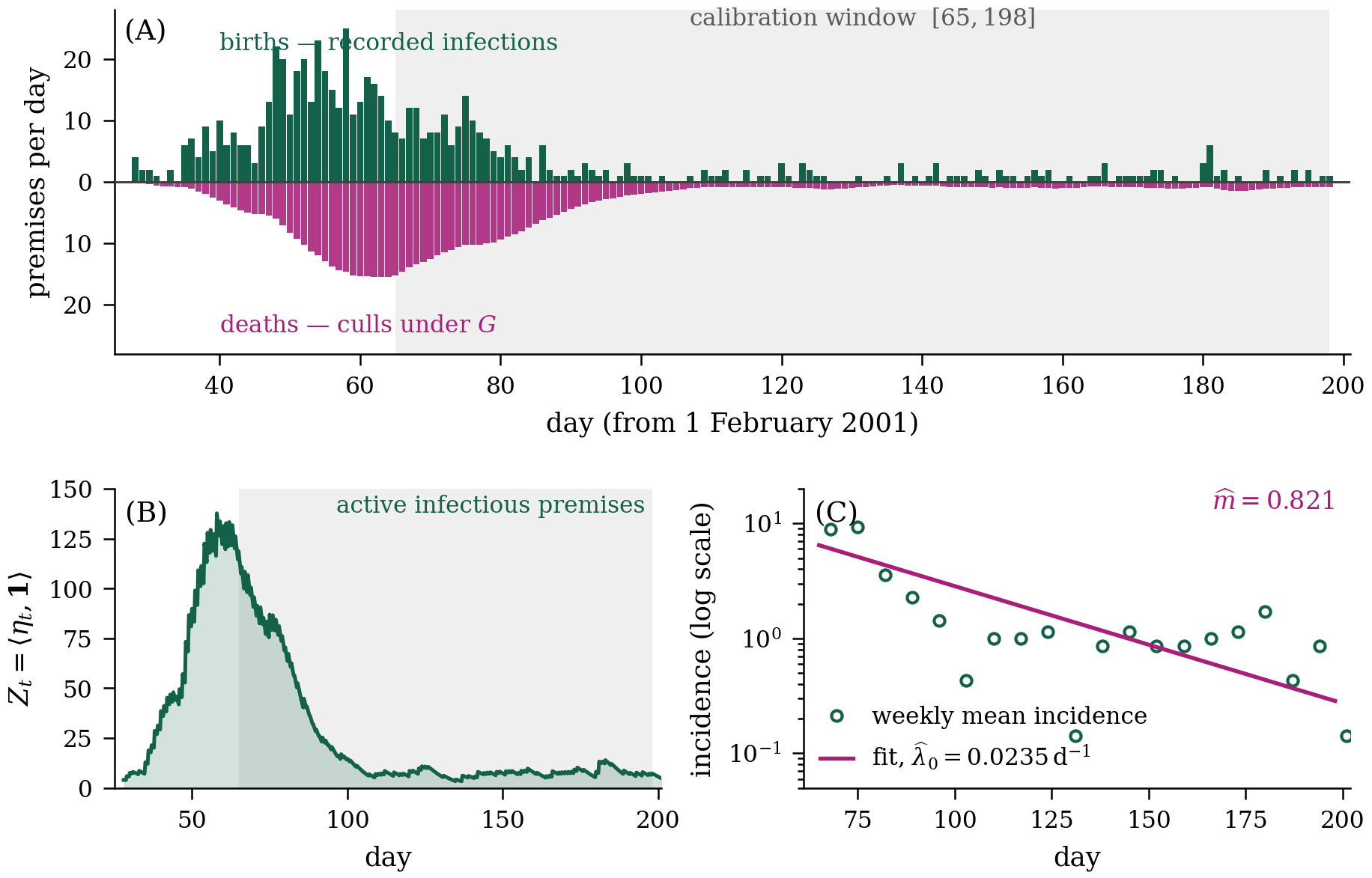}
  \vspace{-0.9cm}
  \caption{The 2001 north Cumbrian foot-and-mouth outbreak read as a branching
    process, \textsf{stpp} data set~\cite{gabriel2013,diggle2005}. (A)~Births,
    the $648$ recorded premises, above the axis; deaths, the culls placed under
    the fitted lifetime law $G=\Gamma(2,4)$, below it; the shaded band is the
    calibration window $[65,198]$, which carries $267$ of the $648$ premises.
    (B)~The reconstructed prevalence $Z_t$, peaking near $140$ active premises
    around day~$58$. (C)~The decline in logarithmic coordinates, weekly mean
    incidence against the fitted exponential; the slope is
    $\widehat\lambda_0=0.0235\,\mathrm{d}^{-1}$, whence $\widehat m=0.821$
    by~\eqref{eq:EL}. Day~$0$ is 1~February 2001.}
  \label{fig:record}
\end{figure}

\begin{figure}[h!]
  \centering
  \includegraphics[width=\linewidth]{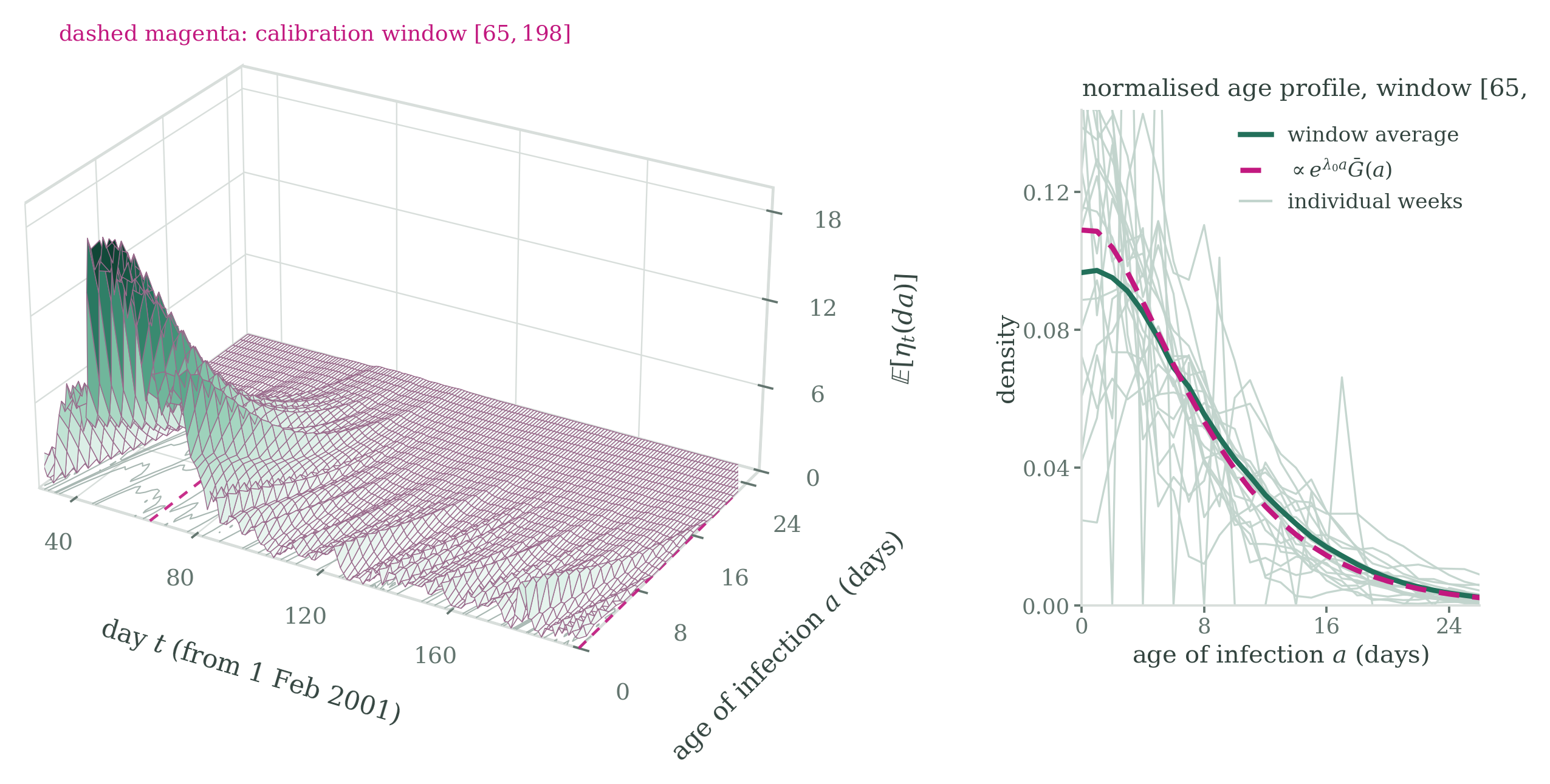}
    \vspace{-0.5cm}
  \caption{The age configuration reconstructed from the record. Left: the mass
    of $\eta_t$ above the time--age plane, with its level curves projected on
    the base and the calibration window marked by the dashed verticals. Each
    recorded premises contributes an atom that ages at unit speed and is removed
    under $G$. The ridge is narrow and stands at low ages, three quarters of the
    mass being carried by premises infectious for less than ten days, and it
    falls in height through the calibration window without moving in age, which
    is the decline of a subcritical population at fixed age structure. Right:
    the same profiles normalised over the window, week by week in grey and in
    average in green, against the stable-age density
    $\propto e^{\lambda_0a}\bar G(a)$. The surface is not available from the
    surveillance record, which contains dates alone.}
  \label{fig:config}
\end{figure}

\subsection{Calibration and verification of the assumptions}
\label{sec:app-calib}

The lifetime law is taken from the published analysis of the full DEFRA record.
Chis~Ster and Ferguson~\cite{chisster2007} report infection-to-cull delays of
some $8.8$--$9.4$ days on the $2\,026$ national premises; netting out a one-day
latent period leaves a mean infectious period near eight days, so we set
$\mathbb{E}[T]=8$ and take $G=\Gamma(k,\vartheta)$ with $k=2$ and
$k\vartheta=\mathbb{E}[T]$, hence $\vartheta=4$ days, the shape being warranted
by the meta-analysis of Mardones et al.~\cite{mardones2010}. This settles two
assumptions outright. The Gamma density is continuous, positive on $(0,\infty)$
and bounded away from zero on compact subsets of it, with $G(0)=0$ and full
support, which is {\rm(H1)}; and its hazard
$\mu(a)=\frac{a}{\vartheta(\vartheta+a)}$ increases from zero to $1/\vartheta$,
so that $\mu^{*}=\mu_\infty=1/\vartheta=0.25$ per day, both finite and positive,
which is {\rm(H4)}. Epidemiologically, a culling pressure that rises with the
age of an infection and saturates is what one expects of a detection process: a
premises infectious for a fortnight is far more likely to have been found than
one infectious for a day, but no regime detects instantly. For the offspring law
we follow Bessell et al.~\cite{bessell2010} and take
$\xi\sim\mathrm{Poisson}(m)$, whence $p_0=e^{-m}$ and $p_1=me^{-m}>0$, the
latter being {\rm(H3)} and being what the survival comparison of
Lemma~\ref{lem:survival} requires. The three remaining assumptions depend on
$m$, hence on the calibration. Equation~\eqref{eq:malthus}, the root condition
$m\int_0^\infty e^{\lambda t}\,dG(t)=1$, is the Euler--Lotka equation; for the
Gamma law it inverts in closed form,
\begin{equation}
  \label{eq:EL}
  \widehat m=\bigl(1-\widehat\lambda_0\,\vartheta\bigr)^{k}.
\end{equation}
A log-linear Poisson regression of the daily incidence over the controlled
decline $[65,198]$, over $134$ days of which $43$ carry no case and all enter
the fit, gives $\widehat\lambda_0=0.0235\ \text{d}^{-1}$ and, through
\eqref{eq:EL}, $\widehat m=0.821$. The residual deviance is $242.2$ on $132$
degrees of freedom and the Pearson dispersion is $\widehat\phi=2.17$, so the
daily counts are not Poisson about the fitted curve: the contacts of an epidemic
are spatially structured, which inflates the variance, and the decline is
steeper in its first weeks than in its long tail, so that a single exponential
is a first-order summary. We therefore treat $\widehat\lambda_0$ as a
quasi-likelihood estimate, unaffected by the dispersion since it solves the same
estimating equation, and widen its interval by $\widehat\phi^{1/2}$, obtaining
$95\%$ confidence intervals $[0.0178,0.0291]$ for $\widehat\lambda_0$ and
$[0.781,0.863]$ for $\widehat m$. Those cover the overdispersion but not the
serial dependence a branching mechanism induces, and are therefore a lower bound
on the uncertainty.

What this delivers is an effective reproduction number for the controlled phase
as a whole, averaged over the window. Its endpoints matter: across start days in
$[60,75]$ and end days in $[180,198]$, $\widehat m$ ranges over $[0.751,0.872]$,
drifting upwards as the window is moved later. A decline that flattens in its
tail admits two readings these data do not separate, residual importation and
heterogeneity of transmission, under which the best-connected premises are
removed first. What matters for what follows is a weaker and verifiable
property, namely that every value $\widehat m$ takes over that family of
windows, and over the confidence interval, leaves the assumptions of
Section~\ref{sec:main-results} in force.

Since the Gamma hazard increases, $\mu^{*}=\mu_\infty$ and {\rm(H5)} reduces to
$p_0<\tfrac12$, the scale $\vartheta$ cancelling. For any offspring law of mean
$m$ one has $\xi\ge1$ on $\{\xi>0\}$, so that $m=\mathbb{E}[\xi]\ge1-p_0$; thus
{\rm(H5)} forces $m>\tfrac12$ whatever the offspring law, and the theory of
Section~\ref{sec:main-results} lives in the barely subcritical regime,
$R_{\mathrm{eff}}$ just below one. That regime is a feature of the lift rather
than a defect of the estimates (Proposition~\ref{prop:sharp}), and it is
precisely the regime in which an outbreak persists long enough for $\nu_Y$ to be
the pertinent reference, and in which surveillance is operationally hard. Under
the Poisson law {\rm(H5)} reads $m>\log2$, which through~\eqref{eq:EL} is a
ceiling on the mean infectious period at fixed decline rate,
\begin{equation}
  \label{eq:H5-ET}
  \mathbb{E}[T]<\frac{k\bigl(1-(\log2)^{1/k}\bigr)}{\widehat\lambda_0}
  =14.3\ \text{days},
\end{equation}
against the eight days used here; the ceiling still stands at $11.5$ days at the
upper end of the confidence interval, and a latent period longer than the one
day netted out above only lowers $\mathbb{E}[T]$ further. At the calibrated
point $p_0=0.440$ and $\|\kappa\|_\infty=p_0\mu^{*}=0.110$, against
$(1-p_0)\mu_\infty=0.140$: the killing rate is beaten by the age-confinement
rate, which is {\rm(H5)} in the form used throughout. At the least favourable
window, $\widehat m=0.751$ gives $p_0=0.472$, still below one half, and
subcriticality {\rm(H2)} holds a fortiori.

Because {\rm(H5)} constrains $p_0$ rather than $m$, it is worth recording how
much latitude the offspring law retains at the calibrated mean. The universal
bound $p_0\ge1-m$ gives $p_0\ge0.179$, and Poisson gives $0.440$, so the
condition admits laws appreciably more dispersed than Poisson, but not
arbitrarily so. Within the negative binomial family of mean $m$ and dispersion
$r$ one has $p_0=(1+m/r)^{-r}$, which crosses one half at $r^{*}=2.10$: the
calibration is robust to a variance-to-mean ratio of up to about $1.4$, and no
further. Since transmission of foot-and-mouth disease is heterogeneous across
premises, this is the assumption most exposed to the choice of a parametric
family, and the one a re-analysis should test first.

It remains to exhibit an exponent $\alpha>1$ with
$\mathbb{E}[\xi^{\alpha}]<\infty$ and $\alpha\lambda_0>\|\kappa\|_\infty$, which
is {\rm(H6)}. The Poisson law has moments of every order, so only the second
requirement binds and it reads $\alpha>\alpha^{*}$, where
\begin{equation}
  \label{eq:alphastar}
  \alpha^{*}:=\frac{\|\kappa\|_\infty}{\lambda_0}
  =\frac{p_0\mu^{*}}{\lambda_0}=4.69
\end{equation}
at the calibrated point, the admissible exponents forming the open interval
$(\alpha^{*},\infty)$. The requirement is mild across the whole calibration,
$\alpha^{*}$ staying below $6.32$ over the confidence interval and over every
window examined above, so a single exponent serves throughout. All six
assumptions therefore hold on the calibrated model.

The branching model also treats the study region as closed, whereas premises
could in principle be infected from outside it. The national movement ban of
late February 2001 makes such importation unlikely over the calibration window,
which opens more than a month later. Its direction is in any case known:
residual importation would flatten the observed decline, lowering
$\widehat\lambda_0$ and raising $\widehat m$, hence lowering $p_0$, so the
margin reported above for {\rm(H5)} is generous rather than conservative.

By Theorem~\ref{thm:consistency} the particle system is ergodic for every
$N\ge2$, so no admissibility threshold restricts the simulation and $N$ is
governed by accuracy alone (Remark~\ref{rem:rate}). What the calibration does
quantify is the cost of the alternative. Estimating $\nu_Y$ by rejection
sampling requires simulating the unconditioned process and retaining the
trajectories still alive at a horizon long enough for the conditional laws to
have converged, and $\mathbb{P}(\tau>T)\asymp e^{-\lambda_0T}$: over the burn-in
of $800$ days used in Section~\ref{sec:app-results} the proxy
$e^{-800\lambda_0}$ is of order $10^{-8}$, so more than $10^{8}$ trajectories
would have to be simulated for one to be retained, while the particle system
keeps its $1500$ particles alive. The cost rises steeply as $m$ decreases
towards $\log2$, reaching some $10^{14}$ at the boundary of {\rm(H5)}. With
{\rm(H1)}--{\rm(H6)} in force, Lemma~\ref{lem:smallsets} and
Lemma~\ref{lem:survival} complete Assumption~H of~\cite{CV2021}, and
Theorem~\ref{thm:convergence} applies: the stationary empirical measure
$\mathcal{X}^N$ estimates $\nu_Y$ with error at most
$d\,N^{-\varpi}\|\varphi\|_\infty$, with $\varpi<\tfrac12$.

\subsection{Quasi-stationary estimation and change detection}
\label{sec:app-results}

We run the particle system with $N=1500$ particles from a single newborn
premises; after a burn-in of $800$ days, functionals of $\nu_Y$ are estimated as
time averages over a stationary window of $8200$ days, sampled every $4$ days,
with standard errors obtained by batch means. Two functionals are of interest
(Figure~\ref{fig:qsd}). The first is the count marginal $\pi_Y(n)=\nu_Y(Z=n)$,
the law of the number of premises simultaneously infectious under persistence,
of mean $\mathbb{E}_{\nu_Y}[Z]=3.008\pm0.011$ farms. The second is the age
profile
\begin{equation}
  \label{eq:ageprofile}
  \bar\alpha_Y(da)
  =\frac{1}{\nu_Y(\langle\eta,\mathbf{1}\rangle)}
   \int_E\eta(da)\,\nu_Y(d\eta),
\end{equation}
the size-biased law of the ages of infection carried by the active premises. It
is the functional the lift to $\mathcal{M}_p(\mathbb{R}_+)$ was built to reach:
not a function of the population size, and not recoverable from $\pi_Y$.

The age profile admits a comparison with classical theory. In the stable-age
theory of age-structured populations the normalised profile is proportional to
$e^{\lambda_0a}\bar G(a)$, of mean $6.731$ days for $G=\Gamma(2,4)$ at the
calibrated $\lambda_0$. The particle system returns $6.733$ days, and the two
densities agree to $0.00037$ in supremum norm for a peak of $0.108$. Nothing in
the simulation encodes the stable-age form: the particles carry ages and are
resampled at killing events, and the profile emerges from that mechanism alone.
The agreement is a check on the simulation rather than an independent
confirmation of the theory, the stable-age density being itself a consequence of
quasi-stationarity.

\begin{figure}[htbp]
  \centering
  \includegraphics[width=0.9\linewidth]{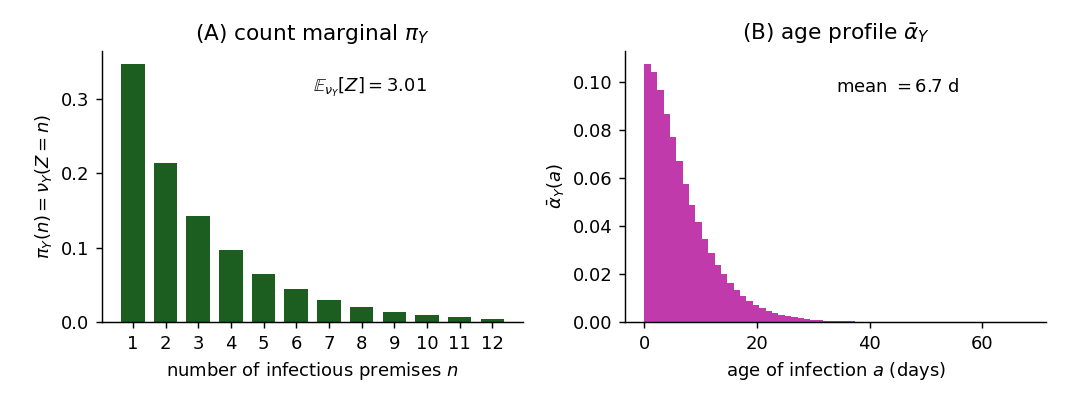}
    \vspace{-0.9cm}
  \caption{The Yaglom limit of the calibrated outbreak, estimated by the
    Fleming--Viot system ($N=1500$). (A)~Count marginal $\pi_Y(n)=\nu_Y(Z=n)$,
    of mean $\mathbb{E}_{\nu_Y}[Z]=3.008$ premises. (B)~Age profile
    $\bar\alpha_Y$, the size-biased law of the ages of the active premises, of
    mean $6.73$ days.}
  \label{fig:qsd}
\end{figure}

The second check is parameter-free and bears on a different stratum.
Corollary~\ref{cor:mean-killing} yields an exact identity: the killing rate being
carried by the singletons,
\begin{equation}
  \label{eq:audit}
  \lambda_0=\nu_Y(\kappa)
  =p_0\,\pi_Y(1)\,\mathbb{E}_{\nu_Y}\!\bigl[\mu(A)\mid Z=1\bigr],
\end{equation}
where $A$ denotes the age of the unique individual alive on the event
$\{Z=1\}$. Its right-hand side is an output of the particle system, while the
left is the rate estimated from the incidence data in
Section~\ref{sec:app-calib} and never supplied to the simulation as such. On the
calibrated model the right-hand side evaluates to $0.023451\pm0.000081$ against
$\lambda_0=0.023477$, a relative discrepancy of $0.11\%$ and a third of one
standard error. The identity tests two results at once. That any
deterministic limit point of the stationary empirical measure is
quasi-stationary with decay rate $\nu^{*}(\kappa)$ is
Theorem~\ref{thm:consistency}(iii), established independently of~\cite{CV2021};
that the selected limit is $\nu_Y$ and its decay rate the Malthusian parameter
is Theorem~\ref{thm:convergence}(ii). Neither is used in building the
simulation, and the two enter through different channels, the first through the
killing rate of the singleton stratum and the second through the Euler--Lotka
root, so the agreement constrains the pair. The rate of Theorem~\ref{thm:convergence}(iii) is confirmed as
well: run over $N$ from ten to three hundred particles, the error
$\mathbb{E}\,|\mathcal{X}^N(\varphi)-\nu_Y(\varphi)|$ for
$\varphi=\mathbf{1}_{\{Z=1\}}$ decays with empirical log-log slope
$-0.504\pm0.005$, so the parametric $N^{-1/2}$ of Remark~\ref{rem:rate} is
attained and the guaranteed exponent $\varpi<\tfrac12$ is not tight.

That part of the theorem holds for every $N\ge2$, whereas
\cite[Theorem~2.3]{CV2021} requires $N>\lambda_1/(\lambda_1-\|\kappa\|_\infty)$;
by Proposition~\ref{prop:sharp} the largest admissible $\lambda_1$ is
$(1-p_0)\mu_\infty$, so that threshold is at best $N^{*}=4.7$ at the calibrated
point, and it diverges as {\rm(H5)} approaches its boundary. The distinction is
not formal: the system is ergodic and estimable at $N=2$, where the right-hand
side of~\eqref{eq:audit} evaluates to $0.0278$ against $\lambda_0=0.0235$, and
the discrepancy falls to $12\%$ at $N=3$, $5\%$ at $N=8$ and $2\%$ at $N=32$.
Accuracy alone governs the choice of $N$, as Remark~\ref{rem:rate} asserts, and
no threshold forbids the small values.

The lift now bears on the surveillance question. A degradation of control that
shows in a slower removal of infected premises, that is, a lengthening of the
infectious period, $G=\Gamma(2,\vartheta)$ with $\vartheta$ rising above its
calibrated value at fixed $m$, ages the active premises while leaving
their number unchanged. The invariance is exact, and it is a change of time
unit: if $T\sim\Gamma(k,1)$ then $\vartheta T\sim\Gamma(k,\vartheta)$, so ages
and lifetimes are dilated by $\vartheta$ while the clock is slowed by the same
factor, and the process is the image of the one at $\vartheta=1$ under the
dilation of the ages. That dilation preserves the number of atoms, so $\pi_Y$
does not depend on $\vartheta$ at all, whereas $\bar\alpha_Y$ is dilated with
it; the rate scales as $\lambda_0=(1-m^{1/k})/\vartheta$ by~\eqref{eq:EL}. Running the system at $\mathbb{E}[T]=11$ days with $m$
unchanged confirms this: none of $\pi_Y(1),\dots,\pi_Y(6)$ moves by more than
$1.1$ batch-means standard errors, and $\mathbb{E}_{\nu_Y}[Z]$ by $0.6$.
Figure~\ref{fig:detection} makes the separation concrete: at a matched
in-control false-alarm rate of $0.10$, a test on $\bar\alpha_Y$ reaches
detection power $0.21$ when the mean infectious period lengthens by three days,
while a test on $\pi_Y$ remains at the false-alarm level. Surveillance is
censored at extinction: the clock runs on $\{Z_t\ge1\}$, each replicate resolves
into alarm, extinction before alarm, or horizon reached, and the quantity
reported is
$\mathbb{P}(\text{alarm}<\min(\text{extinction},\text{horizon}))$, matched in
its false-alarm rate but not in its censoring, the decay rate being smaller
under the alternative.

\begin{figure}[h!]
  \centering
  \includegraphics[width=\linewidth]{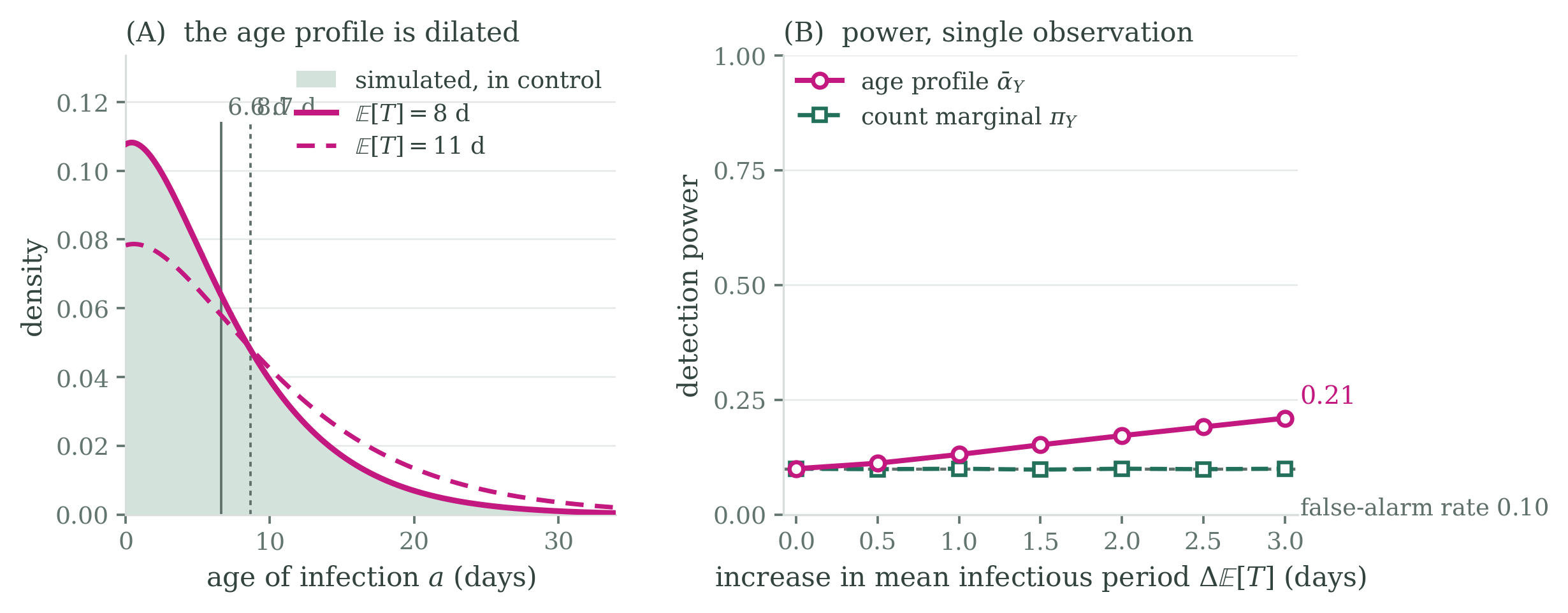}
    \vspace{-0.9cm}
  \caption{Detecting a lengthening of the infectious period ($m$ fixed at the
    calibrated value). (A)~The mechanism: the stable-age density is dilated by
    $\vartheta'/\vartheta$, which moves the mean age while leaving the number of
    premises unchanged; the histogram is the simulated in-control profile.
    (B)~The resulting power, at a matched false-alarm rate of $0.10$, for a
    single observation of the configuration.}
  \label{fig:detection}
\end{figure}

The power is modest, and the reason is structural: under persistence the
outbreak carries three premises on average, so a single observation supplies
about three ages. What the figure establishes is a separation rather than an
operational alarm, and a separation of a particular kind, the count marginal
having no power against this alternative at any sample size. One boundary of
that claim should be stated. The count is blind to the quasi-stationary law of
the prevalence, read from a snapshot; it is not blind to the dynamics, since
$\lambda_0$ falls from $0.0235$ to $0.0171\ \text{d}^{-1}$ under the three-day
change, which is visible in the very series used for calibration. The two
readings are complementary: the rate requires a long stretch of incidence data
and returns one number for the whole window, whereas the age profile is a
property of the current state.

Taken together, the application establishes four things on a real, calibrated
epidemic model. The hypotheses {\rm(H1)--(H6)} hold at the calibrated point and
across every window examined in Section~\ref{sec:app-calib}, in the transparent
form~\eqref{eq:H5-ET} that the mean infectious period stay below some fourteen
days, so the theory of Section~\ref{sec:main-results} is not vacuous but
describes an operationally relevant regime. The quasi-stationary age
configuration, absent from the surveillance record, is reconstructed from a
single estimated rate and a literature lifetime law. The reconstruction is
checked against the stable-age density of classical demography and,
independently, against the parameter-free identity~\eqref{eq:audit}, which it
meets to a third of a standard error, and it attains the parametric rate. And
the age profile detects a slower-culling degradation to which the population
count is exactly blind at quasi-stationarity. That is the concrete payoff of
carrying the Yaglom limit on $\mathcal{M}_p(\mathbb{R}_+)$ rather than on the
integers.
\section{Conclusion}
\label{sec:conclusion}

We have constructed a Fleming--Viot particle system whose particles are age
configurations, and shown that its stationary empirical measure selects the
Yaglom limit of a subcritical Bellman--Harris process, in total variation in
time and at a polynomial rate in the number of particles. Three estimates carry
the argument, none of them inherited from the point-valued theory: an additive
Lyapunov function confining masses and ages at once, an explicit Doeblin
minorisation carried by the singleton stratum, and a direct comparison of
survival probabilities. The drift condition {\rm(H5)} is not an artefact of
these estimates but a necessary feature of the lift
(Proposition~\ref{prop:sharp}), and the Yaglom limit is itself strengthened
along the way, holding in total variation, at an exponential rate, with decay
rate the Malthusian parameter.

The application shows what the lift buys. On the 2001 Cumbrian foot-and-mouth
record, a lifetime law taken from the literature and a single rate read from the
incidence fix the model, and the particle system then regenerates the entire age
configuration, an object the surveillance record does not contain. The
reconstruction is checked twice: against the stable-age density of classical
demography, which the simulation reproduces, and against the parameter-free
identity of Corollary~\ref{cor:mean-killing}, which it meets to a third of a
standard error and which constrains the Malthusian parameter and the decay rate
of the selected limit together. A slower cull ages the active premises without
changing their number, so the count marginal is blind to the degradation at
quasi-stationarity, while the age profile, available only on
$\mathcal{M}_p(\mathbb{R}_+)$, registers it.

\bibliographystyle{elsarticle-num}
\bibliography{biblio}


\appendix

\section{Proof of Lemma~\ref{lem:malthus}}
\label{app:malthus}

\begin{proof}
(i) Fix a scalar $\theta<\mu_\infty$ and a scalar
$\varsigma\in(0,\mu_\infty-\theta)$. By definition of $\mu_\infty$ there is a
scalar $A_\varsigma\ge0$, depending on $\varsigma$ alone, such that
$\mu\ge\mu_\infty-\varsigma$ on $[A_\varsigma,\infty)$, so that
\begin{equation}
\label{eq:tail-Gbar}
\bar G(a)=\exp\Bigl(-\int_0^a\mu\Bigr)
\le\exp\Bigl(-\int_{A_\varsigma}^a\mu\Bigr)
\le e^{(\mu_\infty-\varsigma)A_\varsigma}\,
e^{-(\mu_\infty-\varsigma)a},\; a\ge A_\varsigma.
\end{equation}
Since $\E[e^{\theta T}]=1+\theta\int_0^\infty e^{\theta a}\bar G(a)\,da$ and
$\theta<\mu_\infty-\varsigma$, the integral converges
by~\eqref{eq:tail-Gbar}.

(ii) Let $\Xi:[0,\mu_\infty)\to\R$ be the function
$\Xi(\lambda):=m\,\E[e^{\lambda T}]$, which by (i) is finite, continuous and
strictly increasing there, the last property because $T>0$ almost surely
by {\rm(H1)}. One has $\Xi(0)=m<1$ by {\rm(H2)}. Set the scalar
$\theta_0:=(1-p_0)\mu_\infty$, so that $\theta_0<\mu_\infty$; note also that
$\theta_0<\mu^*$, since $p_0\ge1-m>0$ as recorded in
Section~\ref{sec:prelim} and $\mu_\infty\le\mu^*$. Since $\mu\le\mu^*$
by {\rm(H4)} we have $\bar G(a)\ge e^{-\mu^*a}$, so $T$ stochastically dominates
an $\mathrm{Exp}(\mu^*)$ variable and, $\theta_0<\mu^*$ making the right-hand
side finite, $\E[e^{\theta_0T}]\ge\mu^*/(\mu^*-\theta_0)$, whence
\[
\Xi(\theta_0)\ \ge\ \frac{m\,\mu^*}{\mu^*-\theta_0}\ >\ 1
\;\Longleftrightarrow\;
\theta_0>(1-m)\mu^* ,
\]
and the latter holds because
$(1-m)\mu^*\le p_0\mu^*<(1-p_0)\mu_\infty=\theta_0$, the first inequality
because $p_0\ge1-m$ and the second by {\rm(H5)}. The intermediate value theorem
applied to $\Xi$ on $[0,\theta_0]$ now gives a root $\lambda_0\in(0,\theta_0)$
of~\eqref{eq:malthus}, unique since $\Xi$ increases wherever it is finite.
This proves both the existence of $\lambda_0$ and the strict inequality
$\lambda_0<(1-p_0)\mu_\infty$.

(iii) Write $v=m\,e^{-\lambda_0a}\,J/\bar G$ with
$J(a)=\int_a^\infty e^{\lambda_0s}g(s)\,ds$, so that $J'=-e^{\lambda_0a}g$ and
$\bar G'=-g$. Differentiating and using $g=\mu\bar G$,
\begin{equation}
\label{eq:v-ode}
v'=-\lambda_0v+\frac{m\,e^{-\lambda_0a}}{\bar G}
\Bigl(-e^{\lambda_0a}g+\frac{J\,g}{\bar G}\Bigr)
 =-\lambda_0v+\mu\,(v-m).
\end{equation}
Evaluating~\eqref{eq:reproductive-value} at $a=0$ gives
$v(0)=m\,\E[e^{\lambda_0T}]=1$, which is exactly~\eqref{eq:malthus};
substituting~\eqref{eq:v-ode} and $v(0)=1$
into~\eqref{eq:generator-linear} yields
$\mathcal Av=v'+\mu(m\,v(0)-v)=-\lambda_0v$, and hence
$\mathcal L\langle\eta,v\rangle=\langle\eta,\mathcal Av\rangle
=-\lambda_0\langle\eta,v\rangle$.

For the bounds, substitute $s=a+w$ in~\eqref{eq:reproductive-value} and use
$g=\mu\bar G$:
\begin{equation}
\label{eq:v-integral}
v(a)=m\int_0^\infty e^{\lambda_0w}\,
\exp\Bigl(-\int_a^{a+w}\mu\Bigr)\,\mu(a+w)\,dw.
\end{equation}
As $\supp(G)=[0,\infty)$ by {\rm(H1)} and $\mu_\infty>0$ by {\rm(H4)}, the
residual lifetime is almost surely finite, so
$\int_0^\infty e^{-\int_a^{a+w}\mu}\mu(a+w)\,dw=1$ for every $a$; since
$e^{\lambda_0w}\ge1$, \eqref{eq:v-integral} gives $v\ge m$. For the upper bound,
fix a scalar $\varsigma\in(0,\mu_\infty-\lambda_0)$, which is possible by (ii)
since $\lambda_0<(1-p_0)\mu_\infty\le\mu_\infty$, and let $A_\varsigma$ be as
in (i), that constant depending on $\varsigma$ alone. For $a\ge A_\varsigma$ one
has $\int_a^{a+w}\mu\ge(\mu_\infty-\varsigma)w$, so by~\eqref{eq:v-integral} and
$\mu\le\mu^*$, $v(a)\le m\mu^*\int_0^\infty
e^{-(\mu_\infty-\varsigma-\lambda_0)w}\,dw
=m\mu^*/(\mu_\infty-\varsigma-\lambda_0)<\infty$, while for $a\le A_\varsigma$,
using $m\int_a^\infty e^{\lambda_0s}\,dG(s)\le m\,\E[e^{\lambda_0T}]=1$
in~\eqref{eq:reproductive-value}, $v(a)\le e^{-\lambda_0a}/\bar G(a)
=\exp(\int_0^a(\mu-\lambda_0))\le e^{\mu^*A_\varsigma}$.
Thus $m\le v\le v_{\max}<\infty$, where the scalar $v_{\max}$ is the larger of
the two bounds just obtained. Finally $v'=(\mu-\lambda_0)v-m\mu$
by~\eqref{eq:v-ode}, which is bounded since $\mu$ is continuous and bounded
by {\rm(H4)}; hence $v\in C_b^1(\R_+)$.
\end{proof}
\section{Proof of Lemma~\ref{lem:lyapunov}}
\label{app:lyapunov}

\begin{proof}
Let $\eta\in E$ carry the atoms $a_1,\dots,a_n$ and write
$u:=\langle\eta,v\rangle$, $r:=\langle\eta,\rho\rangle$ and
$n:=\langle\eta,\ind\rangle$, so that $V=1+u^\alpha+\varepsilon r$
by~\eqref{eq:lyapunov}; note that $u\ge mn$ and $r\ge n$, since $v\ge m$ by
Lemma~\ref{lem:malthus}(iii) and $\rho\ge1$. Throughout, the extension
of~\eqref{eq:generator-linear} to the unbounded weight $\rho$ is legitimate by
the localisation argument recorded. Since $V(\mathbf 0)=1$, the killing correction in~\eqref{eq:LX} reads
$\kappa\,(V-V(\mathbf 0))=\kappa\,(u^\alpha+\varepsilon r)$, so that
\begin{equation}
\label{eq:LXV-split}
\mathcal L_XV(\eta)=\Bigl(\mathcal L u^\alpha(\eta)
+\kappa(\eta)\,u(\eta)^\alpha\Bigr)
+\varepsilon\Bigl(\mathcal L r(\eta)+\kappa(\eta)\,r(\eta)\Bigr).
\end{equation}
The two brackets are estimated in turn. We begin with the mass term. By Lemma~\ref{lem:malthus}(iii),
$\mathcal Lu=-\lambda_0u$. The death of individual $i$, which occurs at rate
$\mu(a_i)$, changes $u$ by the random variable
$\Delta_i:=\xi_i\,v(0)-v(a_i)=\xi_i-v(a_i)$, where $\xi_1,\dots,\xi_n$ are
independent copies of $\xi$, and $|\Delta_i|\le\xi_i+v_{\max}$ by
Lemma~\ref{lem:malthus}(iii), so that the jumps of $u$ are bounded in law,
uniformly in $\eta$. This is what the boundedness of $v$ achieves, and it is
what fails if the age weight is placed inside the power. Reading the jumps
off~\eqref{eq:generator},
\begin{equation}
\label{eq:Lu-alpha}
\mathcal L\,u^\alpha
=\alpha u^{\alpha-1}\,\mathcal Lu+R(\eta)
=-\alpha\lambda_0\,u^\alpha+R(\eta),
\;
R(\eta):=\sum_{i=1}^{n}\mu(a_i)\,\E\bigl[(u+\Delta_i)^\alpha-u^\alpha
-\alpha u^{\alpha-1}\Delta_i\bigr] ,
\end{equation}
and $R\ge0$ by the convexity of $x\mapsto x^\alpha$ on $\R_+$ ($\alpha>1$
by {\rm(H6)}) and by $u+\Delta_i\ge0$. There is a finite constant $c_\alpha$,
depending on $\alpha$ alone, such that for $\alpha\in(1,2]$
$0\le(x+y)^\alpha-x^\alpha-\alpha x^{\alpha-1}y\le c_\alpha|y|^\alpha$, and for
$\alpha>2$ the same difference is at most
$c_\alpha(x^{\alpha-2}y^2+|y|^\alpha)$, for all $x\ge0$ and $x+y\ge0$. Set the
finite scalars $M_\alpha:=2^{\alpha-1}(\E[\xi^\alpha]+v_{\max}^\alpha)$ and
$M_2:=2(\E[\xi^2]+v_{\max}^2)$, which bound $\E|\Delta_i|^\alpha$ and
$\E|\Delta_i|^2$ uniformly in $i$ and in $\eta$; both are finite by {\rm(H6)},
the second being needed only when $\alpha>2$, in which case $\E[\xi^2]<\infty$
follows from $\E[\xi^\alpha]<\infty$. Using $\mu\le\mu^*$ by {\rm(H4)} and
$n\le u/m$,
\begin{equation}
\label{eq:R-bound}
0\le R(\eta)\ \le\ \mu^*c_\alpha
\bigl(M_\alpha\,n+M_2\,n\,u^{\alpha-2}\bigr)
\ \le\ C_1\bigl(u+u^{\alpha-1}\bigr)\ =\ o\bigl(u^\alpha\bigr),
\end{equation}
with the convention that the second term inside the bracket is absent when
$\alpha\le2$, with $C_1:=\mu^*c_\alpha(M_\alpha+M_2)/m$ a finite scalar, and
with the last equality holding because $\alpha>1$. As for the killing term in
the first bracket of~\eqref{eq:LXV-split}, it is supported on the singletons
by~\eqref{eq:kappa}, where $u=v(a)$ and
\begin{equation}
\label{eq:K1}
\kappa(\delta_a)\,v(a)^\alpha\ \le\ \|\kappa\|_\infty\,v_{\max}^\alpha
\ =:\ K_1\ <\ \infty.
\end{equation}
Combining~\eqref{eq:Lu-alpha}, \eqref{eq:R-bound} and~\eqref{eq:K1},
\begin{equation}
\label{eq:mass-drift}
\mathcal L u^\alpha+\kappa\,u^\alpha
\ \le\ -\alpha\lambda_0\,u^\alpha+C_1\bigl(u+u^{\alpha-1}\bigr)+K_1.
\end{equation}

We turn to the age term, which is linear in $\eta$ and therefore produces no
convexity remainder. From~\eqref{eq:generator-linear} and $\rho(0)=1$,
$\mathcal A\rho(a)=(\beta-\mu(a))\rho(a)+m\,\mu(a)$. The killing term
$\kappa(\eta)\,r(\eta)$ equals $p_0\mu(a)\rho(a)$ on $\eta=\delta_a$ and
vanishes otherwise by~\eqref{eq:kappa}; using $-\mu\le-(1-p_0)\mu$ on the
configurations where it vanishes, we obtain in all cases
$\mathcal L r(\eta)+\kappa(\eta)\,r(\eta)
\le\sum_{i=1}^{n}(\beta-(1-p_0)\mu(a_i))\rho(a_i)+m\,\mu^*n$.
Choose a scalar $A_0\ge0$ so that $\mu(a)\ge\mu_\infty-\varepsilon'/(1-p_0)$ for
$a\ge A_0$, which is possible by the definition of $\mu_\infty$; then, for
$a\ge A_0$, one has
$\beta-(1-p_0)\mu(a)\le\beta-(1-p_0)\mu_\infty+\varepsilon'=-\zeta$, where the
scalar $\zeta:=(1-p_0)\mu_\infty-\varepsilon'-\beta$ satisfies
$\zeta>\|\kappa\|_\infty>0$, its positivity and this inequality following from
the choice of $\beta$ and $\varepsilon'$ in the statement, which is admissible
by {\rm(H5)}. On the compact $[0,A_0]$ both $\rho$ and $\mu$ are bounded, so
that $(\beta-(1-p_0)\mu(a))\rho(a)\le-\zeta\rho(a)+C_2$ there, for a finite
scalar $C_2$ depending on $\beta$, $\zeta$, $p_0$ and $A_0$ only. Hence, using
again $n\le u/m$,
\begin{equation}
\label{eq:age-drift}
\varepsilon\Bigl(\mathcal Lr+\kappa\,r\Bigr)
\ \le\ -\varepsilon\zeta\,r+\varepsilon\,C_3\,n
\ \le\ -\varepsilon\zeta\,r+\frac{\varepsilon C_3}{m}\,u ,
\; C_3:=C_2+m\mu^*.
\end{equation}
Inserting~\eqref{eq:mass-drift} and~\eqref{eq:age-drift}
into~\eqref{eq:LXV-split},
\[
\mathcal L_XV(\eta)\ \le\ -\alpha\lambda_0\,u^\alpha
-\varepsilon\zeta\,r
+C_1\bigl(u+u^{\alpha-1}\bigr)+\frac{\varepsilon C_3}{m}\,u+K_1.
\]
Since $\alpha>1$, Young's inequality absorbs the three middle terms: for the
scalar $\theta>0$ fixed in the statement there is a finite scalar
$C_4=C_4(\theta)$ with
$C_1(u+u^{\alpha-1})+\varepsilon C_3u/m\le\theta u^\alpha+C_4$. Therefore
\[
\mathcal L_XV(\eta)\ \le\ -(\alpha\lambda_0-\theta)\,u^\alpha
-\zeta\,\varepsilon r+(K_1+C_4)
\ \le\ -\lambda_1\bigl(u^\alpha+\varepsilon r\bigr)+(K_1+C_4)
\ =\ -\lambda_1V(\eta)+C ,
\]
with $\lambda_1=\min(\alpha\lambda_0-\theta,\zeta)$ and
$C:=K_1+C_4+\lambda_1$, a scalar which is therefore explicit, since
$K_1=\|\kappa\|_\infty v_{\max}^\alpha$ by~\eqref{eq:K1} and $C_4$ is furnished
by Young's inequality, the last equality using
$V=1+u^\alpha+\varepsilon r$ once more. This is~\eqref{eq:lyap-global}. That
$\lambda_1>\|\kappa\|_\infty$ holds by the two choices made in the statement:
$\alpha\lambda_0-\theta>\|\kappa\|_\infty$ because
$\theta<\alpha\lambda_0-\|\kappa\|_\infty$, a non-empty range by {\rm(H6)}; and
$\zeta>\|\kappa\|_\infty$ as recorded above. It remains to describe the sublevel sets. If $V(\eta)\le L$ then
$u\le L^{1/\alpha}$ by~\eqref{eq:lyapunov}, hence
$n\le\lfloor L^{1/\alpha}/m\rfloor=R_L$, the mass being an integer; and
$\varepsilon e^{\beta a_i}\le\varepsilon r\le L$ for each atom, hence
$a_i\le\beta^{-1}\log(L/\varepsilon)=A_L$. Thus $K_L$ is contained in the image
of the compact set $\bigsqcup_{k=1}^{R_L}[0,A_L]^k$ under the continuous map
$(a_1,\dots,a_k)\mapsto\sum_{i\le k}\delta_{a_i}$, hence is relatively compact
in the narrow topology, with all limit points carrying at least one atom, so
lying in $E$. Conversely, a narrowly compact subset of $E$ has bounded masses,
since $\eta\mapsto\langle\eta,\ind\rangle$ is narrowly continuous, and, being
tight with integer masses, a common compact support; so $V$ is bounded on it,
and every compact subset of $E$ lies in some $K_L$.
\end{proof}

\section{Proof of Lemma~\ref{lem:smallsets}}
\label{app:smallsets}

\begin{proof}
By the last statement of Lemma~\ref{lem:lyapunov}, every $\eta\in K_L$ carries
$n\le R_L$ individuals, of ages $a_1,\dots,a_n\in[0,A_L]$. Write $R:=R_L$,
$A:=A_L$. Fix scalars $s_1,s_2>0$, put $t_L:=s_1+s_2$ and $A':=A+s_1$, and
choose scalars $0<c_1<c_2<s_2$. We bound $P^{\kappa}_{t_L}(\eta,\cdot)$ from
below by retaining a single trajectory, over two successive time intervals.

On $[0,s_1]$ we reduce the configuration to a singleton. Require individual $1$
to survive and every other individual to die childless. Individual $1$ survives
with probability $\bar G(a_1+s_1)/\bar G(a_1)\ge\bar G(A+s_1)>0$, the last
inequality because $a_1\le A$ and $\bar G$ is decreasing. Individual
$i\in\{2,\dots,n\}$ dies childless in $[0,s_1]$ with probability
$p_0\bigl(G(a_i+s_1)-G(a_i)\bigr)/\bar G(a_i)\ge p_0\,c_{s_1}$, where
\begin{equation}
\label{eq:cs1}
c_{s_1}:=\inf_{a\in[0,A]}\int_{a+s_1/2}^{a+s_1}g
\ \ge\ \frac{s_1}{2}\,\inf_{[s_1/2,\,A+s_1]}g\ >\ 0
\end{equation}
by {\rm(H1)}, the interval $[s_1/2,A+s_1]$ being a compact subset of
$(0,\infty)$. These requirements concern disjoint subtrees, hence are
independent by the branching property, so their joint probability is at least
the scalar $c_*:=\bar G(A+s_1)\,(p_0c_{s_1})^{R-1}>0$, using $p_0c_{s_1}\le1$
and $n-1\le R-1$. At time $s_1$ the configuration is the singleton $\delta_b$
with $b=a_1+s_1\in[s_1,A']$.

On $[s_1,t_L]$ we fix the age. Starting from $\delta_b$, require the individual
to die at time $s_1+w$ for some $w\in(0,s_2)$, producing exactly one offspring,
which is possible since $p_1>0$ by {\rm(H3)}, and require that offspring to
survive until $t_L$, at which time it has age $z=s_2-w$. Applying the strong
Markov property at time $s_1$, the law of $\eta_{t_L}$ dominates the image,
under $z\mapsto\delta_z$, of the measure on $(0,s_2)$ with density
$q_b(z)=[g(b+s_2-z)/\bar G(b)]\,p_1\,\bar G(z)$,
and, for $b\in[s_1,A']$ and $z\in[c_1,c_2]$,
\begin{equation}
\label{eq:b0}
q_b(z)\ \ge\ p_1\,\bar G(c_2)
\inf_{[\,s_1+s_2-c_2,\;A'+s_2-c_1\,]}g\ =:\ b_0\ >\ 0 ,
\end{equation}
again by {\rm(H1)}, the interval being a compact subset of $(0,\infty)$ because
$c_2<s_2$. The population never vanishes along this event, so~\eqref{eq:b0} is a
bound on the killed kernel.

Combining the two intervals through~\eqref{eq:cs1} and~\eqref{eq:b0}, for every
measurable $B\subset E$,
\[
P^{\kappa}_{t_L}(\eta,B)\ \ge\ c_*\,b_0
\int_{c_1}^{c_2}\ind_{\{\delta_{z}\in B\}}\,dz
\ =\ \alpha_L\,\nu_L(B),
\; \alpha_L:=c_*\,b_0\,(c_2-c_1),
\]
where $\nu_L$ is the law of $\delta_{\mathcal U}$ with $\mathcal U$ a random
variable uniform on $[c_1,c_2]$. This is~\eqref{eq:doeblin}. Since
$P^{\kappa}_{t_L}(\eta,E)\le1$ we have $\alpha_L\in(0,1]$, and
$t_L,\alpha_L,\nu_L$ depend only on $(G,h)$ and $L$. Finally
$P_{t_L}\ge P^{\kappa}_{t_L}$, so $K_L$ is a small set for $X$; and by the last
statement of Lemma~\ref{lem:lyapunov} every compact subset of $E$ is contained
in some $K_L$.
\end{proof}

\section{Proof of Lemma~\ref{lem:survival}}
\label{app:survival}

\begin{proof}
Let $K$ be compact, so that $K\subset K_L$ for some scalar $L\ge1$ by the last
statement of Lemma~\ref{lem:lyapunov}; write $R:=R_L$ and $A:=A_L$, and let
$S:\R_+\times\R_+\to[0,1]$ be given by $S(a,t):=\Prob_{\delta_a}(\tau>t)$, which
is non-increasing in $t$.

We first reduce to a single lineage. Let $\eta\in K$ carry the individuals of
ages $a_1,\dots,a_n$ with $n\le R$ and $a_i\le A$, by
Lemma~\ref{lem:lyapunov}. The population survives to $t$ as soon as the subtree
of any single individual does, and it dies out only if all subtrees do; by the
branching property and a union bound,
\begin{equation}
\label{eq:survival-sandwich}
\max_{i\le n}S(a_i,t)\ \le\ \Prob_\eta(\tau>t)\ \le\ \sum_{i\le n}S(a_i,t)
\ \le\ R\,\max_{i\le n}S(a_i,t).
\end{equation}
It therefore suffices to compare $S(a,t)$ and $S(a',t)$ for $a,a'\in[0,A]$,
uniformly in $t$.

Suppose first $a\le a'\le A$. If $t<a'-a$, the individual started at $\delta_a$
is still alive at time $t$ as soon as it survives to age $a+t\le a'$, so
$S(a,t)\ge\bar G(a')/\bar G(a)\ge\bar G(A)\ge\bar G(A)\,S(a',t)$, the last step
because $S(a',t)\le1$. If $t\ge a'-a$, starting from $\delta_a$, the individual
survives to age $a'$ with probability $\bar G(a')/\bar G(a)\ge\bar G(A)$, at
which time the configuration is $\delta_{a'}$; by the Markov property and the
monotonicity of $S(a',\cdot)$,
\begin{equation}
\label{eq:survival-up}
S(a,t)\ \ge\ \frac{\bar G(a')}{\bar G(a)}\,S\bigl(a',t-(a'-a)\bigr)
\ \ge\ \bar G(A)\,S(a',t).
\end{equation}
In both cases $S(a,t)\ge\bar G(A)\,S(a',t)$. Suppose now $a>a'$, and fix a scalar $s>0$. Starting from $\delta_a$, require the individual to die at an age in $(a+s/2,a+s)$ with exactly one offspring, and that offspring to survive to age $a'$. Using {\rm(H3)} for the single offspring
and {\rm(H1)} for the density, the probability of this event is at least
$p_1\bar G(a')\bigl[G(a+s)-G(a+s/2)\bigr]/\bar G(a)
\ge p_1\bar G(A)(s/2)\inf_{[s/2,\,A+s]}g=:\sigma_1>0$,
a scalar depending on $(G,h)$, $A$ and $s$ only, and on that event the
configuration is $\delta_{a'}$ at some time before $s+a'$. Again by the Markov
property and the monotonicity of $S(a',\cdot)$,
\begin{equation}
\label{eq:survival-down}
S(a,t)\ \ge\ \sigma_1\,S(a',t).
\end{equation}

Setting the scalar $\sigma_A:=\min(\bar G(A),\sigma_1)>0$,
\eqref{eq:survival-up} and~\eqref{eq:survival-down} give
$S(a,t)\ge\sigma_A\,S(a',t)$ for all $a,a'\in[0,A]$ and all $t\ge0$. Combined
with~\eqref{eq:survival-sandwich}, this yields, for all $\eta,\eta'\in K$ with
respective ages $(a_i)$ and $(a'_j)$ and all $t\ge0$,
$\Prob_\eta(\tau>t)\ge\max_iS(a_i,t)\ge\sigma_A\max_jS(a'_j,t)
\ge(\sigma_A/R)\,\Prob_{\eta'}(\tau>t)$, whence
$\inf_{t\ge0}\inf_K\Prob_\cdot(\tau>t)/\sup_K\Prob_\cdot(\tau>t)\ge\sigma_A/R>0$.
\end{proof}

\section{Proof of Lemma~\ref{lem:stability}}
\label{app:stability}

\begin{proof}
Put $f:=e^{\psi}$, so that the scalar $c:=\inf f=e^{\inf\psi}$ lies in $(0,1]$
and $f\le1$, since $\psi\le0$ and $\psi$ is bounded. For a bounded measurable
function $f$ on $\R_+$ let $Q_tf(a):=\E_{\delta_a}\bigl[\prod_if(a_i(t))\bigr]$,
which defines a semigroup of operators on bounded measurable functions of the
age; the point of the proof is that it maps $C_b^1(\R_+)$ into itself. By the
branching property,
$\E_\eta\bigl[\prod_if(a_i(t))\bigr]=\prod_j(Q_tf)(a_j)$ for
$\eta=\sum_j\delta_{a_j}$; since $\varphi_\psi(\mathbf 0)=0$, the indicator
$\ind_{\{\tau>t\}}$ is redundant and
$P_t^{\kappa}\varphi_\psi(\eta)=\E_\eta[\varphi_\psi(\eta_t)]
=1-\prod_jQ_tf(a_j)$. It therefore suffices to check that
$\psi_t:=\log Q_tf$ lies in $C_b^1(\R_+)$ and is non-positive. For the
bounds, $Q_tf\le1$ because $f\le1$, so $\psi_t\le0$. From below,
$\prod_if(a_i(t))\ge c^{Z_t}$, and $x\mapsto c^{x}$ being convex, Jensen's
inequality together with the moment bound $\E_{\delta_a}[Z_t]\le1$ recorded
before~\eqref{eq:sup-Z} gives
\begin{equation}
\label{eq:Qtf-lower}
Q_tf(a)\ \ge\ \E_{\delta_a}\bigl[c^{Z_t}\bigr]
\ \ge\ c^{\,\E_{\delta_a}[Z_t]}\ \ge\ c\ >\ 0 ,
\end{equation}
uniformly in $a$, since $c\le1$. For the regularity, decompose on the residual
lifetime of the ancestor and write $u_s:=Q_sf(0)$, a scalar lying in $[c,1]$
by~\eqref{eq:Qtf-lower}, so that $Q_tf(a)=[\bar G(a+t)/\bar G(a)]f(a+t)
+\int_0^t[g(a+s)/\bar G(a)]h(u_{t-s})\,ds$.
Substituting $y=a+s$ moves the dependence on $a$ to the endpoints:
\begin{equation}
\label{eq:Qtf-shifted}
\bar G(a)\,Q_tf(a)=\bar G(a+t)\,f(a+t)
+\int_a^{a+t}g(y)\,h\bigl(u_{t+a-y}\bigr)\,dy.
\end{equation}
Taking $a=0$ in~\eqref{eq:Qtf-shifted} yields the Volterra equation
$u_t=\phi(t)+\int_0^tg(y)\,h(u_{t-y})\,dy$, where $\phi:=\bar Gf$ is $C^1$ with
$\|\phi'\|_\infty\le\mu^*+\|f'\|_\infty<\infty$, since $\bar G'=-g$,
$g=\mu\bar G\le\mu^*$ by {\rm(H4)} and $f\in C_b^1$; recall also that $h$ is
$C^1$ on $[0,1]$ with $h'$ non-decreasing, so that $0\le h'\le h'(1)=m<1$ by
{\rm(H2)}, and that $G\le1$. The map $t\mapsto u_t$ is then $C^1$ with
$U_t:=\sup_{s\le t}|u_s'|<\infty$. This does not follow at once from the
continuity of $g$, which is not assumed differentiable, and we argue in two
steps.

First, $u$ is Lipschitz. Fix $T>0$ and, for $0\le t\le t+\delta\le T$, split the
integral at $y=t$:
\begin{equation}
\label{eq:u-increment}
u_{t+\delta}-u_t=\phi(t+\delta)-\phi(t)
+\int_t^{t+\delta}\!\!g(y)\,h\bigl(u_{t+\delta-y}\bigr)\,dy
+\int_0^t\!g(y)\bigl[h\bigl(u_{t+\delta-y}\bigr)-h\bigl(u_{t-y}\bigr)\bigr]dy .
\end{equation}
Set $\omega(\delta):=\sup\{|u_{t+\delta}-u_t|:\,t+\delta\le T\}$, a quantity
bounded by $1$ since $u$ takes its values in $[c,1]$. Using
$h(u)\in[0,1]$, $g\le\mu^*$, $h'\le m$ and $\int_0^tg=G(t)\le1$, the three terms
of~\eqref{eq:u-increment} are bounded respectively by
$\|\phi'\|_\infty\delta$, $\mu^*\delta$ and $m\,\omega(\delta)$, whence
$\omega(\delta)\le(\|\phi'\|_\infty+\mu^*)\delta+m\,\omega(\delta)$ and
therefore $\omega(\delta)\le L\delta$ with
$L:=(\|\phi'\|_\infty+\mu^*)/(1-m)$, the denominator being positive by
{\rm(H2)}. In particular $u$ is differentiable almost everywhere with
$U_t\le L<\infty$.

Second, $u'$ is continuous. Being Lipschitz, $u$ is absolutely continuous, and
so is $h\circ u$, with $(h\circ u)'=h'(u)\,u'$ almost everywhere. For any
absolutely continuous $F$ on $[0,T]$ one has $F(t)=F(0)+\int_0^tF'$, so that
Fubini gives $(g*F)(t)=F(0)G(t)+(G*F')(t)$; since $G$ is $C^1$ with $G'=g$ and
$G(0)=0$, and $F'\in L^\infty$, the right-hand side is differentiable in $t$
with derivative $F(0)g(t)+(g*F')(t)$, which is continuous because $g$ is
continuous and bounded. Applying this to $F=h\circ u$, whose value at $0$ is
$h(u_0)=h(f(0))$, and adding $\phi'$,
\begin{equation}
\label{eq:u-derivative}
u_t'=\phi'(t)+h\bigl(f(0)\bigr)\,g(t)
+\int_0^t g(y)\,h'\bigl(u_{t-y}\bigr)\,u'_{t-y}\,dy ,
\end{equation}
the right-hand side being continuous in $t$. Hence $u\in C^1$, with $U_t\le L$
by the first step.

Differentiating~\eqref{eq:Qtf-shifted} in $a$, which is legitimate because the
integrand is continuous and $u$ is $C^1$, and using $\bar G'=-g$, $g=\mu\bar G$
and $u_0=f(0)$, we obtain
\begin{align}
\label{eq:Qtf-derivative}
(Q_tf)'(a)&=\mu(a)\Bigl(Q_tf(a)-h(u_t)\Bigr)
+\frac{\bar G(a+t)}{\bar G(a)}\,f'(a+t)+\frac{g(a+t)}{\bar G(a)}\Bigl(h\bigl(f(0)\bigr)-f(a+t)\Bigr)\notag\\
&\;+\frac{1}{\bar G(a)}\int_a^{a+t}\!\!g(y)\,h'\bigl(u_{t+a-y}\bigr)\,
u'_{t+a-y}\,dy.
\end{align}
Each term is bounded uniformly in $a$: $\mu\le\mu^*$ by {\rm(H4)}, and
$Q_tf,\,h(u_t)\in[0,1]$; $\bar G(a+t)/\bar G(a)\le1$ and $f'$ is bounded;
$g(a+t)/\bar G(a)=\mu(a+t)\,\bar G(a+t)/\bar G(a)\le\mu^*$, while
$h(f(0)),\,f(a+t)\in[0,1]$; and the last term is at most $m\,U_t$, since
$h'\le m$ on $[0,1]$ and $\int_a^{a+t}g\le\bar G(a)$. Hence
$Q_tf\in C_b^1(\R_+)$, and $\psi_t=\log Q_tf\in C_b^1(\R_+)$
by~\eqref{eq:Qtf-lower}, since $Q_tf\ge c>0$.
\end{proof}
\end{document}